\documentclass{llncs}

\usepackage{amsmath}
\usepackage{amssymb}
\usepackage{amsfonts}
\usepackage{amscd}
\usepackage{amsxtra}
\usepackage{mathrsfs}
\usepackage{nicefrac}
\usepackage{bbold}
\usepackage[shortlabels]{enumitem}

\newlength{\hspaceforlengthglumpf}

\newcommand{\onespace}{\mspace{1mu}}

\newcommand{\comment}[1]{\text{\footnotesize[#1]}}

\newcommand{\eqcmt}[1]{\mathrel{\mathop=\limits_{#1}}}
\newcommand{\lecmt}[1]{\mathrel{\mathop\le\limits_{#1}}}

\DeclareMathOperator{\supp}{supp}
\DeclareMathOperator{\argmax}{argmax}

\DeclareMathOperator{\polylog}{polylog}

\newcommand{\Zero}{\mathbf{0}}

\newcommand{\lt}{\left}
\newcommand{\rt}{\right}

\newcommand{\abs}[1]{{\lt\lvert{#1}\rt\rvert}}

\newcommand{\sabstight}[2][2]{{\lvert\mspace{-#1mu}{#2}\mspace{-#1mu}\rvert}}

\newcommand{\nfrac}[2]{{\nicefrac{#1}{#2}}}

\newcommand{\NN}{\mathbb{N}}

\DeclareMathOperator*{\Prb}{\mathbf{P}}
\DeclareMathOperator*{\Exp}{\mathbf{E}}

\DeclareMathOperator*{\Var}{\mathbf{Var}}

\DeclareMathOperator{\IndicatorOp}{\mathbf{I}}
\newcommand{\Ind}{\IndicatorOp}

\newcommand{\eps}{\varepsilon}

\newlength{\algotabbingwidth}
\numberwithin{theorem}{section}

\newcommand{\mypar}{\par\medskip\noindent}

\newcommand{\todo}[1][ToDo]{{\tiny\color{red}$\scriptscriptstyle\circ$}%
  {\marginpar{\centering\tiny\sf todo$\color{red}\bullet$\color{red}#1}}}

\newcommand{\Mark}[1]{\todo[\color{green}Mark]\bgroup\markoverwith{\textcolor{green}{\rule[-0.5ex]{2pt}{0.4pt}}}\ULon{#1}}

\newcommand{\longversion}[1]{}
\DeclareMathOperator{\ndccOP}{\mathsf{N}}
\DeclareMathOperator{\rcOP}{\mathsf{C}}
\DeclareMathOperator{\frcOP}{\mathsf{C^*}}
\DeclareMathOperator{\foolOP}{\mathsf{F}}

\DeclareMathOperator{\onerecOP}{\mathsf{R^1}}

\newcommand{\ndcc}{\ndccOP}
\newcommand{\rc}{\rcOP}
\newcommand{\frc}{\frcOP}
\newcommand{\fool}{\foolOP}

\newcommand{\onerec}{\onerecOP}

\DeclareMathOperator{\logb}{log}
\renewcommand{\log}[1]{\logb_{#1}}
\DeclareMathOperator{\plainlog}{log}

\newcommand{\notp}{{\bar p}}

\DeclareMathOperator{\Bin}{Bin}

\newcommand{\Matrix}[2]{f}

\begin{document}
\mainmatter              % start of the contributions
\title{Nondeterministic Communication Complexity of random Boolean functions}
\titlerunning{Nondeterministic CC of random Boolean functions}
\toctitle{Nondeterministic Communication Complexity of random Boolean functions}
\author{Mozhgan Pourmoradnasseri\inst{1} \and Dirk Oliver Theis\inst{1}}
\authorrunning{M.Pourmoradnasseri \& D.O.Theis.}
\tocauthor{Mozhgan Pourmoradnasseri and Dirk Oliver Theis}
\institute{%
  University of Tartu\\
  Institute of Computer Science\\
  {\"U}likooli 17\\
  51014 Tartu, Estonia,\\
  \email{\{mozhgan,dotheis\}@ut.ee}\\
  WWW: \texttt{http://ac.cs.ut.ee/}%
}
\maketitle              % typeset the title of the contribution

\begin{abstract}
  We study nondeterministic communication complexity and related concepts (fooling sets, fractional covering number) of random functions $f\colon X\times Y \to \{0,1\}$ where each value is chosen to be~1 independently with probability $p=p(n)$, $n := \abs{X}=\abs{Y}$.
\end{abstract}
%

% CV/intro.tex  (tamc16)
\section{Introduction}
Communication Complexity lower bounds have found applications in areas as diverse as sublinear algorithms, space-time trade-offs in data structures, compressive sensing, and combinatorial optimization (cf., e.g., \cite{Roughgarden:CC-AD:2015,Fiorini-Massar-Pokutta-Tiwary-Dewolf:ACM:15}).
In combinatorial optimization especially, there is a need to lower bound \textsl{nondeterministic} communication complexity~\cite{Yannakakis:91,Kaibel:optima:11} .

Let $X,Y$ be sets and $f\colon X\times Y \to \{0,1\}$ a function.  In nondeterministic communication, Alice gets an $x\in X$, Bob gets a $y\in Y$, and they both have access to a bit string supplied by a prover.  In a protocol, Alice sends one bit to Bob; the decision whether to send 0 or~1 is based on her input~$x$ and the bit string~$z$ given by the prover.  Then Bob decides based on his input~$y$, the bit string~$z$ given by the prover, and the bit sent by Alice, whether to accept (output 1) or reject (output 0).  The protocol is successful, if, (1) regardless of what the prover says, Bob never accepts if $f(x,y)=0$, but (2) for every $(x,y)$ with $f(x,y)=1$, there is a proof~$z$ with which Bob accepts.  The nondeterministic communication complexity is the smallest number~$\ell$ of bits for which there is a successful protocol with $\ell$-bit proofs.

Formally, the following basic definitions are common:
\begin{itemize}
\item The \textit{support} is the set of all \textit{1-entries:} $\supp f := \{ (x,y) \mid f(x,y) = 1 \}$;
\item a \textit{1-rectangle} is a cartesian product of sets of inputs $R = A\times B \subseteq X\times Y$ all of which are 1-entries: $A\times B \subseteq \supp f$;
\item a \textit{cover} (or \textit{1-cover}) is a set of 1-rectangles $\{R_1=A_1\times B_1,\dots, R_k=A_k\times B_k\}$ which together cover all 1-entries of~$f$, i.e., $\bigcup_{j=1}^k R_j = \supp f$;
\item the cover number $\rc(f)$ of $f$ is the smallest size of a 1-cover.
\end{itemize}
One can then define the \textit{nondeterministic communication complexity} simply as $\ndcc(f) := \log{2}\rc(f)$ \cite{Kushilevitz-Nisan:Book:97}.

In combinatorial optimization, one wants to lower bound the nondeterministic communication complexity of functions which are defined based on relations between feasible points and inequality constraints of the optimization problem at hand: Alice has an inequality constraint, Bob has a feasible point, and they should reject (answer~0) if the point satisfies the inequality with equality.

Consider, the following example (it describes the so-called \textit{permuthahedron}).  Let~$k\ge3$ be a positive integer.
\begin{itemize}
\item Let~$Y$ denote the permutations~$\pi$ of $[k]$---the feasible points.
\item Let~$X$ denote the set of non-empty subsets $U\subsetneq [k]$; such an~$U$ corresponds to an inequality constraint $\sum_{u\in U} \pi(u) \ge \abs{U}(\abs{U}+1)/2$.
\end{itemize}
Goemans~\cite{Goemans:permutahedron:15} gave an $\Omega(\plainlog k)$ lower bound for the nondeterministic communication complexity of the corresponding function:
\begin{equation*}
  f(\pi,U) =
  \begin{cases}
    0, & \text{ if $\sum_{u\in U} \pi(u) = \abs{U}(\abs{U}+1)/2$};\\
    1, & \text{ otherwise, i.e., $\sum_{u\in U} \pi(u) > \abs{U}(\abs{U}+1)/2$.}
  \end{cases}
\end{equation*}
For $k=3$, see the following table.  The rows are indexed by the set~$X$, the columns by the set~$Y$.
\begin{equation*}\footnotesize
  \begin{array}{c|cccccc|}
    \multicolumn{1}{c}{}&
    {}        123 & 132 & 213 & 231 & 312 & \multicolumn{1}{c}{321} \\
    \cline{2-7}
    \{1\}   &  0  &  0  &  1  &  1  &  1  &  1  \\
    \{2\}   &  1  &  1  &  0  &  1  &  0  &  1  \\
    \{3\}   &  1  &  1  &  1  &  0  &  1  &  0  \\
    \{1,2\} &  0  &  1  &  0  &  1  &  1  &  1  \\
    \{1,3\} &  1  &  0  &  1  &  0  &  1  &  1  \\
    \{2,3\} &  1  &  1  &  1  &  1  &  0  &  0  \\
    \cline{2-7}
  \end{array}
\end{equation*}

In this situation, the nondeterministic communication complexity lower bounds the logarithm of the so-called \textit{extension complexity:} the smallest number of linear inequalities which is needed to formulate the optimization problem.  This relationship goes back to Yannakakis' 1991 paper~\cite{Yannakakis:91}, and has recently been the focus of renewed attention \cite{Beasley-Klauck-Lee-Theis:Dagstuhl:13,Klauck-Lee-Theis-Thomas:Dagstuhl:15} and a source of some breakthrough results \cite{Fiorini-Massar-Pokutta-Tiwary-Dewolf:lin-vs-semidef:12,Fiorini-Kaibel-Pashkovich-Theis:CombLB:13}.  Other questions remain infamously open, e.g., the nondeterministic communication complexity of the minimum-spanning-tree function: For a fixed number~$k$, Bob has a tree with vertex set $[k]$, Alice has one of a set of inequality constraints (see \cite{Schrijver:Book:03} for the details), and they are supposed to answer~1, if the tree does not satisfy the inequality constraint with equality.

%%%%%%% \fbox{???}\TODO Functions derived from combinatorial optimization problems are often quite ``dirty'': The sets $X$ and~$Y$ are sets of combinatorial objects (e.g., stable vertex sets and cliques in a given graph; trees, matchings, etc), and the function value $f(x,y)$ depends on how these two related in a manner determined by the ``geometry'' of the optimization problem.  It is thus not surprising that the study of random structures has been attempted as a first approximation~\cite[e.g.]{Braun-Fiorini-Pokutta:rndStable:14}.

In this paper, we focus on random functions, and we give tight upper and lower bounds for the nondeterministic communication complexity and its most important lower bounds: the fooling set bound; the ratio number of 1-entries over largest 1-rectangle; the fractional cover number.  For that, we fix $\abs{X}=\abs{Y}=n$, and, we take $f(x,y)$, $(x,y)\in X\times Y$, to be independent Bernoulli random variables with parameter $p=p(n)$, i.e., $f(x,y) = 1$ with probability~$p$ and $f(x,y) = 0$ with probability~$1-p$.

In Communication Complexity, it is customary to determine these parameters up to within a constant factor of the number of bits, but in applications, this is often not accurate enough.  E.g., the above question about the extension complexity of the minimum-spanning-tree polytope asks where in the range between $(1+o(1)) 2\plainlog n$ bits and $(1+o(1))3\plainlog n$ bits the nondeterministic communication complexity lies.  (Here $n$ should taken as $\abs{Y}=2^k-2$.)  Therefore, in our analyses, we focus on the constant factors in our communication complexity bounds.

\subsection{Relationship to related work}
In core (Communication) Complexity Theory, random functions are usually used for establishing that hard functions exist in the given model of computation.  In this spirit, some easy results about the (nondeterministic) communication complexity of random functions and related parameters exist, with $p$ a constant, mostly $p=\nfrac12$ (e.g., the fooling set bound is determined in this setting in \cite{Dietzfelbinger-Hromkovic-Schnitger:96}).

In contrast to this, in applications, the density of the matrices is typically close to~1, e.g., in combinatorial optimization, the number of 0s in a ``row'' $\{y\in Y\mid f(x,y)=0 \}$, is very often polylog of~$n$.  This makes necessary to look at these parameters in the spirit of the study of properties of random graph where $p=p(n)\to 1$ with $n\to\infty$.  In an analogy to the fields of random graphs, the results become both considerably more interesting and also more difficult that way.

The random parameters we analyze have been studied in other fields beside Communication Complexity.  Recently, Izhakian, Janson, and Rhodes~\cite{Izhakian-Janson-Rhodes:PAMS:15} have determined asymptotically the triangular rank of random Boolean matrices with independent Bernoulli entries.  The triangular rank is itself important in Communication Complexity~\cite{Lovasz-Saks:CC-lattice:93} (and its applications~\cite{LeeTheis12}), and it is a lower bound to the size of a fooling set.  In that paper, determining the behavior for $p\to0,1$ is posed as an open problem.

The size of the largest monochromatic rectangle in a random Bernoulli matrix was determined in~\cite{Park-Szpankowski:biclusters:05} when~$p$ is bounded away from 0 and~1, but their technique fails for $p\to 1$.

The nondeterministic communication complexity of a the clique-vs-stable set problem on random graphs was studied in \cite{Braun-Fiorini-Pokutta:rndStable:14}.

%%%%%%%%%%%%%
\mypar%
The parameters we study in this paper are of importance beyond Communication Complexity and its direct applications.  In combinatorics, e.g., the cover number coincides with strong isometric dimension of graphs~\cite{Froncek-Jerebic-Klavzar-Kovar:CPC:07}, and has connections to extremal set theory and Coding Theory~\cite{Hajiabolhassan-Moazami:code:12,Hajiabolhassan-Moazami:cover-free:12}.

The size of the largest monochromatic rectangle is of interest in the analysis of gene expression data~\cite{Park-Szpankowski:biclusters:05}, and formal concept analysis~\cite{Dawande-Keskinocak-Swaminathan-Tayur:2001}.

%%%%%%%%%%%%%

\mypar%
Via a construction of Lov\'asz and Saks~\cite{Lovasz-Saks:CC-lattice:93}, the 1-rectangles, covers, and fooling sets of a function~$f$ correspond to stable sets, colorings, and cliques, resp., in a graph constructed from the function.  Consequently, determining these parameters could be thought of as analyzing a certain type of random graphs.  This approach does not seem to be fruitful, as the probability distribution on the set of graphs seems to have little in common with those studied in random graph theory.  Here is an important example for that.  In the usual random graph models (Erd\H{o}s-Renyi, uniform regular), the chromatic number is within a constant factor of the independence ratio (i.e., the quotient independence number over number of vertices), and, in particular, of the fractional chromatic number (which lies between the two).  The corresponding statement (replace ``chromatic number'' by ``cover number''; ``independence ratio'' by ``Hamming weight of~$f$ divided by the size of the largest 1-rectangle''; ``fractional chromatic number'' by ``fractional cover number'') is false for random Boolean functions, as we will see in Section~\ref{sec:ndcc}.

\paragraph{This paper is organized as follows.}
We determine the size of the largest monochromatic rectangle in Section~\ref{sec:alpha}.  Section~\ref{sec:fool} is dedicated to fooling sets: we give tight upper and lower bounds.  Finally, in Section~\ref{sec:ndcc} we give bounds for both the covering number and the fractional covering number.

% END OF intro.tex

% CV/basics.tex  (tamc16)
\subsection{Definitions}
A Boolean function $f\colon X\times Y\to\{0,1\}$ can be viewed as a matrix whose rows are indexed by~$X$ and the columns are indexed by~$Y$.  We will use the two concepts interchangeably.
In particular, for convenience, we speak of ``row'' $x$ and ``column'' $y$.
We will always take $n=\abs{X}=\abs{Y}$ without mentioning it. Clearly, a \textit{random Boolean function $f\colon X\times Y\to\{0,1\}$ with parameter~$p$} is the same thing as a random $n\times n$ matrix with independent Bernoulli entries with parameter~$p$.

We use the usual conventions for asymptotics: $g \ll h$ and $g=o(h)$ is the same thing.  As usual, $g = \Omega(1)$ means that~$g$ is bounded away from~$0$.  We are interested in asymptotic statements, usually for $n\to \infty$.  A statement (i.e., a family of events $E_n$, $n\in\NN$) holds \textit{asymptotically almost surely, a.a.s.,} if its probability tends to~1 as $n\to \infty$ (more precisely, $\displaystyle \lim_{n\to\infty} \Prb( E_n ) = 1$).

% END OF basics.tex

% CV/alpha.tex  (tamc16)
\section{Largest 1-rectangle}\label{sec:alpha}
As mentioned in the introduction, driven by applications in bioinformatics, the size of the largest monochromatic rectangle in a matrix with independent (Bernoulli) entries, has been studied longer than one might expect.  Analyzing computational data, Lonardi, Szpankowski, and Yang~\cite{Lonardi-Szpankowski-Yang:Conf:04,Lonardi-Szpankowski-Yang:Journal:06} conjectured the shape of the 1-rectangles.  The conjecture was proven by Park and Szpankowski~\cite{Park-Szpankowski:biclusters:05}.  Their proof can be formulated as follows:  Let $f\colon X\times Y \to \{0,1\}$ be a random Boolean function with parameter~$p$.
\begin{itemize}
\item If $\Omega(1) = p \le 1/e$, then, a.a.s., the largest 1-rectangle consists of the 1-entries in a single row or column, and $\onerec(f) = (1+o(1))pn$.
\item If $p \ge 1/e$ but bounded away from~1, then with $a:= \argmax_{b \in\{1,2,3,\dots\}} b p^b$, a.a.s.\ the largest 1-rectangle has~$a$ rows and $p^an$ columns, or vice-versa.
\end{itemize}

The existence of these rectangles is fairly obvious.  Proving that no larger ones exist requires some work.  The problem with the union-bound based proof in \cite{Park-Szpankowski:biclusters:05} is that it breaks down if $p$ tends to~1\ moderately quickly.  In our proofs, we work with strong tail bounds instead.

\paragraph{\bf Our result} extends the theorem in~\cite{Park-Szpankowski:biclusters:05} for the case that $p$ tends to 0 or~1 quickly.

For $K\subseteq X$, the \textit{1-rectangle\label{page:generated_rectangle} of~$f$ generated by $K$} is $R := K\times L$ with
\begin{equation*}
  L := \Bigl\{ y\in Y \mid \forall \; x \in K\colon \ f(x,y)=1 \Bigr\}.
\end{equation*}
The 1-rectangle generated by a subset $L$ of~$Y$ is defined similarly.

\begin{theorem}\label{thm:alpha}
  Let~$f\colon X\times Y\to\{0,1\}$ be a random Boolean function with parameter~$p=p(n)$.
  \begin{enumerate}[(a)]
  \item\label{thm:alpha:small-p} If $\nfrac5n \le p \le \nfrac1e$, then a.a.s., the largest 1-rectangle is generated  by a single row or column, and if $p\gg (\ln n)/n$, its size is $(1+o(1))pn$.
  \item\label{thm:alpha:large-p} Define
    \begin{equation}\label{eq:alpha:large-p:def-a}
      \begin{aligned}
        a_- &:= \lfloor \log{\nfrac1p}e \rfloor, \\
        a_+ &:= \lceil \log{\nfrac1p}e \rceil, \text{ and}\\
        a   &:= \argmax_{b \in \{a_-,a_+\}} b p^b \; = \argmax_{b \in\{1,2,3,\dots\}} b p^b.\\
      \end{aligned}
    \end{equation}
    There exists a constant $\lambda_0$, such that if $\nfrac1e \le p \le 1-\nfrac{\lambda_0}{n}$, then, a.a.s., a largest 1-rectangle is generated by~$a$ rows or columns and its size is $(1+o(1))ap^an$.
  \end{enumerate}
\end{theorem}

The proof requires us to upper bound the sizes of square 1-rectangles, i.e., $R = K\times L$ with $\abs{K}=\abs{L}$.  Sizes of square 1-rectangles have been studied, too.  Building on work in~\cite{Dawande-Keskinocak-Tayur:wp:1996,Dawande-Keskinocak-Swaminathan-Tayur:2001,Park-Szpankowski:biclusters:05}, it was settled in~\cite{Sun-Nobel:JMachLearn:08}, for constant~$p$.  We need results for $p\to 0,1$, but, fortunately, for our theorem, we only require weak upper bounds.

For the proof of~\ref{thm:alpha:small-p}, we say that a 1-rectangle is \textit{bulky}, if it extends over at least~2 rows and also over at least~2 columns.
We then proceed by considering three types of rectangles:
\begin{enumerate}[1.]
\item those consisting of exactly one row or column (they give the bound in the theorem);
\item square bulky rectangles;
\item bulky rectangles which are not square.
\end{enumerate}

For the proof of~\ref{thm:alpha:large-p}, we also require an appropriate notion of ``bulky'': here, we say that a rectangle of dimensions $k\times \ell$ is bulky if $k \le \ell$.  By again considering square rectangles, we prove that a bulky rectangle must have $k < n/\lambda^{\nfrac23}$.  (We always define $\lambda$ through $p = 1 - \nfrac{\lambda}{n}$.)  By exchanging the roles of rows and columns, and multiplying the final probability estimate by~2, we only need to consider 1-rectangles with at least as many columns as rows (i.e., bulky ones).  Following that strategy yields the statement of the theorem.

The complete proof is in Appendix~\ref{apx:alpha}.

\begin{remark}\label{rem:alpha:facts-about-a}
  \begin{enumerate}[(a)]
  \item If $p\ge \nfrac1e$, then
    \begin{equation}\label{eq:alpha:p-to-a-lb}
      \nfrac{1}{e^2}
      \le
      \frac{p}{e}
      \le
      p\cdot p^{\log{\nfrac1p}e}
      \le
      p^a
      \le
      \frac{1}{p}\cdot p^{\log{\nfrac1p}e}
      \le
      \frac{1}{pe}
      \le
      \nfrac{1}{e},
    \end{equation}
    i.e., $p^a \approx \nfrac1e$, more accurately $p^a = (1-o_{p\to 1}(1))/e$.
  \item With $p = 1-\notp = 1 - \nfrac{\lambda}{n}$, the following makes the range of $\onerec(f)$ clearer: Since $\notp \le \ln(\nfrac{1}{(1-\notp)}) \le \notp+\notp^2$ holds when $\notp \le 1-\nfrac1e$, we have
    \begin{equation}\label{eq:alpha:a-bounds}
      \frac{1}{e\bar p} = \frac{n}{e\lambda}   \le\  p\frac{n}{\lambda} = \frac{p}{\bar p}   \le\  \frac{1}{1+\notp}\cdot \frac{1}{\notp} \,\le  \log{\nfrac1p}e  \le  \frac{1}{\notp} =\frac{n}{\lambda}
    \end{equation}
  \end{enumerate}
\end{remark}

\begin{corollary}\label{cor:alpha:large-p:1-o1}
  For $p = 1-\frac{\lambda}{n}$ with $\lambda_0 \le \lambda = o(n)$, we have $\displaystyle%
  \onerec(f)
      % &=
      % \frac{n^2}{e\lambda\,(1+\frac{\lambda}{2n})} + O\bigl(\lambda + \frac{n^{\nfrac32}}{\lambda} \ln n\bigr)\\ &
  =
      \frac{n^2}{e\lambda} + O(n).
      $
\end{corollary}
See Appendix~\ref{apx:alpha} for the proof.

% END OF alpha.tex

% CV/fool.tex  (tamc16)
\newcommand{\Hbipf}{H_f}
\newcommand{\Hbipfnp}{H_f}
\newcommand{\Gnp}[2]{\mathbf{G}_{#1,#2}}

\section{Fooling sets}\label{sec:fool}
A \textit{fooling set} is a subset $F\subseteq X\times Y$ with the following two properties: (1) for all $(x,y)\in F$, $f(x,y) = 1$; and (2) and for all $(x,y),(x',y')\in F$, if $(x,y)\ne(x',y')$ then $f(x,y')f(x',y)=0$.  When~$f$ is viewed as a matrix, this means that, after permuting rows and columns, $F$ identifies the diagonal entries of a submatrix which is~1 on the diagonal, and in every pair of opposite off-diagonal entries, at least one is~0.  We denote by $\fool(f)$ the size of the largest fooling set of~$f$.  The maximum size of a fooling set of a random Boolean function with $p=\nfrac12$ is easy to determine (e.g., \cite{Dietzfelbinger-Hromkovic-Schnitger:96}).

An obvious lower bound to the fooling set size is the \textit{triangular rank,} i.e., the size of the largest triangular submatrix, again after permuting rows and columns.  (There is also an upper bound for the fooling set size in terms of the linear-algebraic rank, cf.~\cite{Dietzfelbinger-Hromkovic-Schnitger:96,FriesenTheis13}, but since our random matrices have high rank, we cannot use that here.)
In a recent Proc.~AMS paper, Izhakian, Janson, and Rhodes \cite{Izhakian-Janson-Rhodes:PAMS:15} determined the triangular rank of a random matrix with independent Bernoulli entries with constant parameter~$p$.  They left as an open problem to determine the triangular rank in the case when $p \to$ 0 or~1, which is our setting.

Our constructions of fooling sets of random Boolean functions make use of ingredients from random graph theory.  First of all, consider the bipartite $\Hbipf$ whose vertex set is the disjoint union of $X$ and~$Y$, and with $E(\Hbipf) =\supp f \subseteq X$.  For random~$f$, this graph is an \textit{Erd\H{o}s-Renyi random bipartite graph:} each edge is picked independently with probability~$p$.  Based on the following obvious fact, we will use results about matchings in Erd\H{o}s-Renyi random bipartite graphs:

\begin{remark}
  Let $F\subseteq X\times Y$.  The following are equivalent.
  \begin{enumerate}[(a)]
  \item $F$ is a \textit{fooling set}.
  \item $F$ satisfies the following two conditions:
    \begin{itemize}
    \item $F$ is a matching, i.e., $F \subseteq E(H)$;
    \item $F$ is \textit{cross-free,} i.e., for all $(x,y),(x',y')\in F$, if $(x,y)\ne(x',y')$ then $(x,y') \notin E$ or $(x',y)\notin E$.
    \end{itemize}
  \end{enumerate}
\end{remark}

Secondly, fooling sets can be obtained from stable sets in an auxiliary graph: For a random Boolean function~$f$, this graph is an \textit{Erd\H{o}s-Renyi random graphs,} for which results are available yielding good lower bounds.

Fig.~\ref{fig:fool-bounds} summarizes our upper and lower bounds: Upper bounds are above the dotted lines; lower bounds are below the dotted lines; the range for~$p$ is between the dotted lines.  All upper bounds are by the 1st moment method.

We emphasize that the upper and lower bounds differ by at most a constant factor.  If $p\to 1$ quickly enough, i.e., $\notp=1-p = n^{-a}$ for a constant~$a$, then the upper bounds and lower bounds are even the same except for rounding.

\begin{figure}[thp]
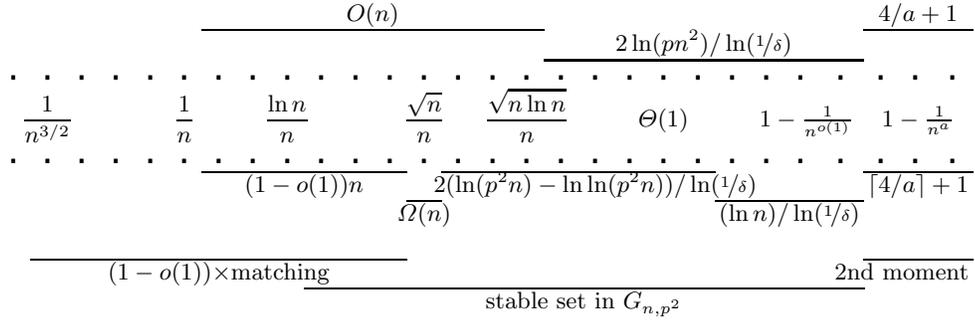

  \centering\small%
  \hspace*{-1em}%
  \begin{tabular}{@{}c@{}c@{}c@{}c@{}c@{}c@{}c@{}c@{}c@{}c@{}c@{}c@{}c@{}c@{}c@{}c@{}c@{}c@{}c@{}c@{}c@{}c@{}c@{}c@{}c@{}c@{}c@{}c@{}}
    \hspace*{1.4em}&\hspace*{1.4em}&\hspace*{1.4em}&\hspace*{1.4em}&\hspace*{1.4em}&\hspace*{1.4em}&\hspace*{1.4em}&\hspace*{1.4em}&\hspace*{1.4em}&\hspace*{1.4em}&\hspace*{1.4em}&\hspace*{1.4em}&\hspace*{1.4em}&\hspace*{1.4em}&\hspace*{1.4em}&\hspace*{1.4em}&\hspace*{1.4em}&\hspace*{1.4em}&\hspace*{1.4em}&\hspace*{1.4em}&\hspace*{1.4em}&\hspace*{1.4em}&\hspace*{1.4em}&\hspace*{1.4em}&\hspace*{1.4em}&\hspace*{1.4em}&\hspace*{1.4em}&\hspace*{1.4em}\\
    &\multicolumn{5}{c}{}&\multicolumn{10}{c}{$O(n)$}&\multicolumn{9}{c}{}&\multicolumn{3}{c}{$\displaystyle 4/a+1$}\\
    \cline{7-16}\cline{26-28}
    \multicolumn{16}{c}{}&\multicolumn{9}{c}{$\displaystyle 2\ln(pn^2)/\ln(\nfrac1\delta)$}\\
    \cline{17-25}
    $\centerdot$&$\centerdot$&$\centerdot$&$\centerdot$&$\centerdot$&$\centerdot$&$\centerdot$&$\centerdot$&$\centerdot$&$\centerdot$&$\centerdot$&$\centerdot$&$\centerdot$&$\centerdot$&$\centerdot$&$\centerdot$&$\centerdot$&$\centerdot$&$\centerdot$&$\centerdot$&$\centerdot$&$\centerdot$&$\centerdot$&$\centerdot$&$\centerdot$&$\centerdot$&$\centerdot$&$\centerdot$\\
    \multicolumn{3}{c}{$\dfrac{1}{n^{3/2}}$}&\quad&\multicolumn{3}{c}{$\dfrac{1}{n}$}&\multicolumn{3}{c}{$\dfrac{\ln n}{n}$}&\quad&\multicolumn{3}{c}{$\dfrac{\sqrt n}{n}$}&\multicolumn{3}{c}{$\dfrac{\sqrt{n\ln n}}{n}$}&\quad&\multicolumn{3}{c}{$\Theta(1)$}&\quad&\multicolumn{3}{c}{$1-\frac{1}{n^{o(1)}}$}&\multicolumn{3}{c}{$1-\frac{1}{n^a}$}\\
    $\centerdot$&$\centerdot$&$\centerdot$&$\centerdot$&$\centerdot$&$\centerdot$&$\centerdot$&$\centerdot$&$\centerdot$&$\centerdot$&$\centerdot$&$\centerdot$&$\centerdot$&$\centerdot$&$\centerdot$&$\centerdot$&$\centerdot$&$\centerdot$&$\centerdot$&$\centerdot$&$\centerdot$&$\centerdot$&$\centerdot$&$\centerdot$&$\centerdot$&$\centerdot$&$\centerdot$&$\centerdot$\\
    \cline{7-12}\cline{14-21}\cline{26-28}
    \multicolumn{6}{c}{}&\multicolumn{6}{c}{$(1-o(1))n$}&\multicolumn{11}{c}{$2(\ln(p^2n) -\ln\ln(p^2n))/\ln(\nfrac1\delta)$}&\multicolumn{2}{c}{}&\multicolumn{3}{c}{$\displaystyle \lceil 4/a\rceil+1$}\\
    \cline{13-13}\cline{22-25}
    \multicolumn{11}{c}{}&\multicolumn{3}{c}{$\Omega(n)$}&\multicolumn{7}{c}{}&\multicolumn{4}{c}{$(\ln n)/\ln(\nfrac1\delta)$}\\
    \\
    \cline{2-12}\cline{26-28}
    &\multicolumn{11}{c}{$(1-o(1))\times$matching}&\multicolumn{11}{c}{}&\multicolumn{5}{r}{2nd moment}\\
    \cline{10-25}
    &\multicolumn{8}{c}{}&\multicolumn{16}{c}{stable set in $G_{n,p^2}$}\\
  \end{tabular}
  \caption{\label{fig:fool-bounds} Upper and lower bounds on fooling set sizes.  ($\delta := 1-p^2$)}
\end{figure}

\subsection{Statement of the theorem, and a glimpse of the proof}
Denote by $\nu(H)$ the size of the largest matching in a bipartite graph~$H$.
For $q=q(m)$, denote by $\Gnp{m}{q}$ the graph with vertex set~$\{1,\dots,m\}$ in which each of the $\binom{m}{2}$ possible edges is chosen (independently) with probability~$q$.  Let $a(q)=a_m(q)$ be a function with the property that, a.a.s., every Erd\H{o}s-Renyi random graph on~$m$ vertices with edge-probability~$q$ has an independent set of size at least~$a_m(q)$.

\begin{theorem}\label{thm:fool}
  Let~$f\colon X\times Y\to\{0,1\}$ be a random Boolean function with parameter~$p=p(n)$.  Define $\notp := 1-p$ and $\delta := 1-p^2$.
  \begin{subequations}
    \begin{enumerate}[(a)]
    \item\label{thm:fool:o-sqrtn}%
      For $n^{-3/2} \le p = o(1/\sqrt n)$, a.a.s., we have
      \begin{equation*}
        \fool(f) = (1-o(1))\nu(\Hbipfnp).
      \end{equation*}
    \item\label{thm:fool:alpha}%
      If $pn - \ln n \to \infty$, then, a.a.s., $\fool(f) \ge a(p^2)$.
    \item\label{thm:fool:ub}%
      If $p \gg \sqrt{(\ln n)/n}$ and $\notp \ge n^{-o(1)}$, then, a.a.s.,
      \begin{equation*}
        \fool(f) \le 2 \log{\nfrac1\delta}(pn^2).
      \end{equation*}
    \item\label{thm:fool:notp-small}%
      If $a \in \lt]0,4\rt[$ is a constant and $\notp = n^{-a}$, then $\displaystyle \fool(f) \le \nfrac{4}{a}+1$.  %%
      If, in addition, $a<1$, then $\displaystyle \fool(f) = \lfloor \nfrac{4}{a} \rfloor+1$
    \end{enumerate}
  \end{subequations}
\end{theorem}
The proof is in Appendix~\ref{apx:fool}.

To obtain the bounds in Fig.~\ref{fig:fool-bounds}, the following facts from random graph theory are needed.
\begin{theorem}[Matchings in Erd\H{o}s-Renyi random bipartite graphs, cf., e.g., \cite{Janson-Luczak-Rucinski:Book}]\label{thm:fool:matchings-in-Hnp}
  Let~$H = (X,Y,E)$ be a random bipartite graph with $\abs{X}=\abs{Y}=n$, and edge probability~$p$.
  \begin{enumerate}[(a)]
  \item If $p \gg \nfrac1n$, then, a.a.s., $H$ has a matching of size $(1-o(1))n$.
  \item If $p = (\omega(n) + \ln n)/n$ for an $\omega$ which tends to~$\infty$ arbitrarily slowly, then, a.a.s, $H$ has a matching of size~$n$.
  \end{enumerate}
\end{theorem}

\begin{theorem}[Stable sets in Erd\H{o}s-Renyi random graphs]\label{thm:fool:stablesets-in-Gnp}
  Let~$G = ([m],E)$ be a random graph with $\{u,v\} \in E$ with edge probability~$q=q(m)$.
  \begin{enumerate}[(a)]
  \item E.g., \cite{Janson-Luczak-Rucinski:Book}: Let $\omega=\omega(m)$ tend to~$\infty$ arbitrarily slowly.  If $\omega/m \le q = 1-\Omega(1)$, then a.a.s., $G$ has a stable set of size at least
    \begin{equation*}
      2\frac{ \ln(qm) - \ln\ln(qm) }{ \ln(1-q) }.
    \end{equation*}
  \item Greedy stable set: If $q = \Omega(1)$, then, a.a.s., $G$ has a stable set of size at least
    \begin{equation*}
      \frac{ \ln(m) }{ \ln(1-q) }.
    \end{equation*}
  \end{enumerate}
\end{theorem}

For the region $p=\Theta(1/\sqrt n)$, there is a corresponding theorem (e.g.,\cite{Dani-Moore:w2ndMoment:11}). We give here an argument about the expectation based on Tur\'an's theorem.  Tur\'an's theorem in the version for stable sets \cite{Alon-Spencer:Book} states that in a graph with vertex set~$V$, there exists a stable set of size at least
\begin{equation*}
  \sum_{v\in V} \frac{1}{\deg(v) + 1},
\end{equation*}
where $\deg(v)$ denotes the degree of vertex~$v$.  For random graphs on vertex set~$V=[m]$ with edge probability $q = c/m$ for a constant~$c$, using Jensen's inequality, we find that there expected size of the largest stable set is at least
\begin{multline*}
  \Exp\lt( \sum_{v\in V} \frac{1}{\deg(v) + 1} \rt)
  =
  \sum_{v\in V} \Exp\lt( \frac{1}{\deg(v) + 1} \rt)
  \\
  \ge
  \sum_{v\in V} \frac{1}{\Exp\deg(v) + 1}
  =
  \frac{2m}{q(m-1) + 1}
  \ge
  \frac{2m}{c + 1}
  =
  \Theta(m).
\end{multline*}

% END OF fool.tex

% CV/ndcc.tex  (tamc16)
\section{Fractional cover number and cover number}\label{sec:ndcc}
Armed with the fooling set and 1-rectangle-size lower bounds, we can now bound the cover number and the fractional cover number.  We start with the easy case $p \le \nfrac12$.

Let~$f$ be a random Boolean function $X\times Y\to\{0,1\}$ with parameter~$p$, as usual.  If $1/n \ll p \le \nfrac12$, we have $\rc(f) = (1-o(1))n$.  Indeed, for $p=o(1/\sqrt n)$, Theorem~\ref{sec:fool}\ref{thm:fool:o-sqrtn} gives the lower bound based on the fooling set lower bound.   For $\nfrac1e \ge p \gg (\ln n) / n)$, Theorem~\ref{thm:alpha}\ref{thm:alpha:small-p} yields $\onerec(f) = (1+o(1))pn$, a.a.s., and for $\nfrac1e \le p \le \nfrac12$, the value of~$a$ in eqn.~\eqref{eq:alpha:large-p:def-a} of Theorem~\ref{thm:alpha}\ref{thm:alpha:large-p} is~1, so that $\onerec(f) = (1+o(1))pn$ there, too.
We conclude that, a.a.s.,
\begin{equation*}
  \rc(f)
  \ge
  \frac{\abs{\supp f}}{\onerec(f)}
  =
  \frac{(1-o(1)) pn^2}{(1-o(1)) pn}
  =
  (1-o(1))\,n.
\end{equation*}

\mypar%
As indicated in the introduction, the case $p > \nfrac12$ is more interesting, both from the application point of view and from the point of view of the proof techniques.

For the remainder of this section, we assume that $p > \nfrac12$.  Define $\notp := 1-p$, and $\lambda := \notp n$.

\subsection{The fractional cover number}
We briefly review the definition of the fractional cover number.  Let~$f$ be a fixed Boolean function, and let~$R$ be a random 1-rectangle of~$f$, drawn according to a distribution~$\pi$.  Define
\begin{equation*}\label{eq:chi:def-min-cover-prb}
  \gamma(\pi) := \min \Bigl\{ \Prb_{R\sim \pi}\bigl( (x,y) \in R \bigr) \mid (x,y)\in \supp f \Bigr\}.
\end{equation*}
The \textit{fractional cover number} is $\frc(f) := \min_\pi 1/\gamma(\pi)$, where the minimum is taken over all distributions~$\pi$ on the set of 1-rectangles of~$f$.
%%It is OK to restrict the supports of the distributions~$\pi$ to a subset~$\mathcal R$ of all 1-rectangles of~$G$, as long as~$\mathcal R$ contains all inclusion-wise maximal rectangles sets.

The following inequalities are well-known~\cite{Kushilevitz-Nisan:Book:97}.
\begin{equation}\label{eq:coloring-parameters}
  \lt.
  \begin{array}[c]{r}
    \displaystyle \frac{\abs{\supp f}}{\onerec(f)} \\[2ex]
    \displaystyle \fool(f)
  \end{array}
  \rt\}
  \le
  \frc(f)
  \le
  \rc(f)
  \lecmt{(*)}
  \bigl( 1+\ln\onerec(f) \bigr) \, \frc(f).
\end{equation}

\subsubsection*{Lower bound}
Theorem~\ref{thm:alpha}\ref{thm:alpha:large-p} allows us to lower bound $\frc(f)$.  Let~$f$ be a random Boolean function $X \times Y\to\{0,1\}$ with parameter~$p>\nfrac12$.
With $\nfrac{\lambda}{n} = \notp = 1-p$, we have a.a.s.,
\begin{equation}\label{eq:chi:lb-ir}
  \frac{\abs{\supp f}}{\onerec(f)}
  \ge
  \frac{(1+o(1))pn^2}{(1+o(1)) n/e\ln(1/p)}
  =
  (1+o(1))\; ep\ln(1/p) n
  \ge
  (1-o(1))\; ep \lambda
\end{equation}
where the last inequality follows from $\notp \le \notp + \notp^2/2 + \notp^3/3 + \dots = \ln(1/(1-\notp))$.
For $\notp=o(1)$, this is asymptotic to $e \lambda$.  It is worth noting that the first inequality in~\eqref{eq:chi:lb-ir} becomes an asymptotic equality if $\notp=o(1)$.

\subsubsection*{Upper bound}
We now give upper bounds on $\frc(f)$.
To prove an upper bound~$b$ on the fractional cover number for a fixed function~$f$, we have to give a distribution~$\pi$ on the 1-rectangles of~$f$ such that, if~$R$ is sampled according to~$\pi$, we have, for all $(x,y)$ with $f(x,y)=1$,
\begin{equation*}
  \Prb( (x,y) \in R ) \ge \nfrac{1}{b}.
\end{equation*}

To prove an ``a.a.s.'' upper bound for a random~$f$, we have to show that
\begin{equation}\label{eq:chi:fracchi-ub:prob_rnd-rect-covers-everything}
  \Prb\biggl( \exists (x,y)\colon\  \Prb\bigl( (x,y) \in R  \mid  f \ \&\  f(x,y)=1 \bigr) < \nfrac1b  \biggr) = o(1).
\end{equation}

Our random 1-rectangle~$R$ within the random Boolean function~$f$ is sampled as follows.  Let $K$ be a random subset of $X$, by taking each $x$ into $K$ independently, with probability~$q$.  Then let $R := K\times L$ be the 1-rectangle generated (see p.~\pageref{page:generated_rectangle}) by the row-set~$K$, i.e., $L := \{ y \mid \forall x \in K\colon f(x,y)=1\}$.

For $y\in Y$, let the random variable $Z_y$ count the number of $x\in X$ with $f(x,y)=0$---in other words, the number of zeros in column~$y$---and set $Z := \max_{y\in Y} Z$.
For $(x,y)\in X\times Y$, conditioned on $f$ and $f(x,y) =1$, the probability that $(x,y) \in R$ equals
\begin{equation*}
  q(1-q)^{Z_y} \ge q(1-q)^Z,
\end{equation*}
so that for every positive integer~$z$, using $\nfrac1b = q(1-q)^z$ in~\eqref{eq:chi:fracchi-ub:prob_rnd-rect-covers-everything},
\begin{equation}\label{eq:chi:fracchi-ub:prb_in-rnd-rect}
  \Prb\biggl( \exists (x,y)\colon\  \Prb\bigl( (x,y) \in R  \mid  f \ \&\  f(x,y)=1 \bigr) < q(1-q)^z  \biggr)
  =
  \Prb( Z > z ).
\end{equation}
To obtain upper bounds on the fractional cover number, we give a.a.s.\ upper bounds on~$Z$, and choose~$q$ accordingly.

\begin{theorem}\label{thm:chi:fracchi-ub}
  \begin{subequations}
    Let $\nfrac12 > p = 1 - \notp = 1 - \nfrac{\lambda}{n}$.
    \begin{enumerate}[(a)]
    \item\label{enum:chi:fracchi-ub:omega(logn)}%
      If $\ln n \ll \lambda < n/2$, then, a.a.s., $\displaystyle%\begin{equation}
      {\textstyle (1-o(1))} \; pe\lambda \le \frc(f) \le {\textstyle (1+o(1))} \; e\lambda $%\end{equation}
    \item\label{enum:chi:fracchi-ub:Theta(logn)}%
      If $\lambda = \Theta(\ln n)$, then, a.a.s., $\displaystyle%\begin{equation}
      \frc(f) = \Theta(\ln n).  $%\end{equation}
    \item\label{enum:chi:fracchi-ub:o(logn)}%
      If $1 \ll \lambda = o(\ln n)$, then, a.a.s.,
      \begin{equation*}
        {\textstyle (1-o(1))} \; \lambda
        \le
        \frc(f)
        \le
        {\textstyle (1+o(1))} \;
        e\max\Bigl(   2\lambda,  \frac{\ln n}{\ln((\ln n)/\lambda)}   \Bigr)
      \end{equation*}
    \end{enumerate}
  \end{subequations}
\end{theorem}
To summarize, we can determine the fractional cover number accurately in the region $\ln n \ll \lambda \ll n$.  For $\lambda = \Theta(\ln n)$ and for $\lambda=\Theta(n)$, we can determine $\frc(f)$ up to a constant.  However, for $\lambda = o(\ln n)$, there is a large gap between our upper and lower bounds.
\begin{proof}
  The lower bounds follow from the discussion above.

  \textit{Proof of the upper bound in \ref{enum:chi:fracchi-ub:omega(logn)}.} %
  For every constant $t>0$, let
  \begin{equation*}
    \psi(t) := 1/\bigl( (1+t) \ln(1+t) - t  \bigr).
  \end{equation*}
  With
  \begin{equation*}
    h(t) = h(t,n) := \frac{\lambda}{\psi(t) \ln n},
  \end{equation*}
  using the a standard Chernoff estimate (Theorem~2.1, Eqn.(2.5) in~\cite{Janson-Luczak-Rucinski:Book}) we find that
  \begin{equation*}
    \Prb\bigl( Z_1 \ge (1+t)\lambda \bigr) \le e^{-\lambda/\psi(t)} \le e^{-h(t)}n,
  \end{equation*}
  so that, by the union bound,
  \begin{equation}\label{eq:chi:fracchi-ub:bound_on_Z_w_h}
    \Prb\bigl( Z \ge (1+t)\lambda \bigr) \le e^{-h(t)}.
  \end{equation}

  For every fixed $t>0$, $h(t)$ tends to infinity with~$n$, so that the RHS in~\eqref{eq:chi:fracchi-ub:bound_on_Z_w_h} is $o(1)$.  Using that in~\eqref{eq:chi:fracchi-ub:prb_in-rnd-rect}, we obtain
  \begin{equation*}
    \Prb\biggl( \exists (x,y)\colon\  \Prb\bigl( (x,y) \in R  \mid  f \ \&\  f(x,y)=1 \bigr) < q(1-q)^{(1+t)\lambda}  \biggr)
    =
    \Prb( Z > (1+t)\lambda )
    = o(1),
  \end{equation*}
  and, taking $q := \frac{1}{(1+t)\lambda}$, we obtain, a.a.s.,
  \begin{equation*}
    \frc(f) \le \frac{1}{q(1-q)^{(1+t)\lambda}} \le \frac{1+t}{1+\frac{1}{(1+t)\lambda}}e\lambda,
  \end{equation*}
  where we used $(1-\eps)^k \ge (1-k\eps^2)e^{-k\eps}$ for $\eps < 1$.
  Since this is true for every~$t>0$, we conclude that, a.a.s., $\frc(f) \le (1-o(1))e\lambda$.

  \mypar%
  \textit{Proof of the upper bounds in \ref{enum:chi:fracchi-ub:Theta(logn)}, \ref{enum:chi:fracchi-ub:o(logn)}.} %
Here we use a slightly different Chernoff bound (Lemma~\ref{lem:binomial-ub} in the appendix).

  For~\ref{enum:chi:fracchi-ub:Theta(logn)}, suppose that $\lambda \le C\ln n$ for a constant $C>1$.  Using Lemma~\ref{lem:binomial-ub} with $\alpha = e^2 C \ln n$, we obtain
  \begin{equation*}
    \Prb\bigl( Z_1 \ge e^2 C \ln n \bigr)
    =
    O\bigl(\nfrac{1}{\sqrt{\ln n}}\bigr)
    e^{-\lambda} \Bigl( \frac{eC\ln n}{e^2 C \ln n} \Bigr)^{\alpha}
    =
    O\bigl(\nfrac{1}{\sqrt{\ln n}}\bigr)
    e^{-\ln n}.
  \end{equation*}
  and thus
  \begin{equation*}
    \Prb\bigl( Z \ge e^2 C\ln n \bigr) = o(1).
  \end{equation*}
  We conclude similarly as above: with $q := \frac{1}{e^2C\ln n}$ we obtain, a.a.s., $\frc(f) \le e^3 C\ln n$.

  Finally, for~\ref{enum:chi:fracchi-ub:o(logn)}, if $\lambda = o(\ln n)$, let $\eps > 0$ be a constant, and use Lemma~\ref{lem:binomial-ub} again, with
  \begin{equation*}
    \alpha := \max\biggl( 2\lambda, \ \frac{ (1+\eps)\ln  n}{  \ln\bigl( \frac{\ln n}{e\lambda} \bigr) } \biggr).
  \end{equation*}
  We find that
  \begin{equation*}
    \Prb\bigl( Z_1 \ge \alpha \bigr)
    =
    o\bigl( e^{- \alpha \ln(\alpha/e\lambda) } \bigr),
  \end{equation*}
  and the usual calculation (Appendix~\ref{ssec:apx:chi:swift-caculation}) shows that $\alpha\ln(\alpha/e\lambda) \ge \ln n$, which implies
  \begin{equation*}
    \Prb\bigl( Z \ge \alpha \bigr) = o(1).
  \end{equation*}
  Conclude similarly as above, with $q := \frac{1}{\alpha}$, we obtain, a.a.s.,
  \begin{equation*}
    \frc(f) \le e\alpha = e\max\lt( 2\lambda, (1+\eps) \frac{\ln n}{  \ln\bigl( \frac{\ln n}{e\lambda} \bigr) } \rt).
  \end{equation*}
  One obtains the statement in the theorem by letting $\eps$ tend to~0; the $e$-factor in the denominator of the $\ln$ of the denominator in~$\alpha$ is irrelevant as $n\to \infty$.
\end{proof}

\subsubsection*{The cover number}
Inequality~($*$) in~\eqref{eq:coloring-parameters} gives us corresponding upper bounds on the cover number.
\begin{corollary}\label{cor:chi:chi-ub-from-fracchi}
  We have $\displaystyle (1-o(1))\, \lambda \le \rc(f)$, and:
  \begin{enumerate}[(a)]
  \item if $\ln n \ll \lambda = O(n/\ln n)$,  then, a.a.s., $\displaystyle \rc(f) = O(\lambda \ln n)$;
  \item if $\lambda = \Theta(\ln n)$,  then, a.a.s., $\displaystyle \rc(f) = O(\ln^2 n)$;
  \item if $1 \ll \lambda = o(\ln n)$, then, a.a.s., $\displaystyle \rc(f) = O\Bigl(
    \max\Bigl(   \lambda\ln n,  \frac{\ln^2 n}{\ln((\ln n)/\lambda)}   \Bigr) \Bigr).
    $%\end{equation}
  \end{enumerate}
  \qed
\end{corollary}

\subsection{Binary-Logarithm of the number of distinct rows, and the ratio $\rcOP/\frcOP$}
When we view~$f$ as a  matrix, the binary logarithm of the number of distinct rows is a lower bound on the cover number of~$f$ \cite{Kushilevitz-Nisan:Book:97}.  We have the following.

\begin{proposition}\label{prop:chi:log-2-lb}\mbox{}
  \begin{enumerate}[(a)]
  \item\label{enum:chi:log-2-lb:large-p}%
    If $\nfrac12 \ge \notp = \Omega(\nfrac1n)$, then, a.a.s., the 2-Log lower bound on $\rc(f)$ is $(1-o(1))\log{2} n$.
  \item\label{enum:chi:log-2-lb:small-p}%
    If $\notp = n^{-\gamma}$ for $1 < \gamma \le \nfrac32$, then a.a.s., the 2-Log lower bound on $\rc(f)$ is $(1-o(1))(2-\gamma)\log{2} n$.
  \end{enumerate}
\end{proposition}
\begin{proof}
  Directly from the following Lemma~\ref{lem:chi:numo_distinct_rows} about the number of distinct rows, with $\lambda = n^{1-\gamma}$.
\end{proof}
\begin{lemma}\label{lem:chi:numo_distinct_rows}\mbox{}
  \begin{enumerate}[(a)]
  \item\label{enum:apx-chi:numo_distinct_rows:large_p}%
    If $\nfrac12 \ge \notp = \Omega(\nfrac1n)$, then, a.a.s., $f$ has $\Theta(n)$ distinct non-zero rows.
  \item\label{enum:apx-chi:numo_distinct_rows:small_p}%
    With $\notp = \lambda/n$, if $n^{-\nfrac12} \le \lambda \le \nfrac12$, then, a.a.s., $f$ has $\Omega(\lambda n)$ distinct non-zero rows.
  \end{enumerate}
  (The constants in the big-Omegas are absolute.)
\end{lemma}
For the sake of completeness, we sketch the proof in Appendix~\ref{ssec:apx:chi:distinct-rows}.

\mypar
Erd\H{o}s-Renyi random graphs have the property that the chromatic number is within a small constant factor from the lower bound one obtains from the independence ratio.
For the cover number of Boolean functions, this is not the case.  Indeed, Theorem~\ref{thm:chi:fracchi-ub}\ref{enum:chi:fracchi-ub:o(logn)}, together with Proposition~\ref{prop:chi:log-2-lb}, shows that, a.a.s.,
\begin{equation*}
  \frac{  \rc(f)  }{ \frc(f)  }
  \ge
  (1+o(1))
  \frac{ \log{2} n }{ \  \frac{ \ln n }{ \ln\bigl(\frac{\ln n}{\lambda}\bigr) } \  }
  =
  \Omega\!\lt( \ln\Bigl( \frac{\ln n}{\lambda} \Bigr) \rt),
\end{equation*}
which is $\Omega(\ln\ln n)$ if $\lambda = \ln^{o(1)}n$.

This gap is more pronounced in the (not quite as interesting) situation when $\lambda =o(1)$.  Consider, e.g., $\lambda = n^{-\eps}$, for some $\eps = \eps(n) = o(1/\ln\ln n)$, say.  Similarly to the proofs of Theorem~\ref{thm:chi:fracchi-ub}, we obtain that $\frc(f) \le e\max(10,2/\eps)$.  (The $\max$-term comes from the somewhat arbitrary upper bound $Z \le \max(10,2/\eps)$.)   For the Log-2 lower bound on the cover number, we have $(1-\eps)\log{2}n$, by Proposition~\ref{prop:chi:log-2-lb}, and thus
\begin{equation*}
  \frac{  \rc(f)  }{ \frc(f)  }
  =
  \Omega( \eps\ln n ).
\end{equation*}

% END OF ndcc.tex

\section{Acknowledgments}
The authors would like to thank the anonymous referees for their valuable comments.

Dirk Oliver Theis is supported by Estonian Research Council, ETAG (\textit{Eesti Teadusagentuur}), through PUT Exploratory Grant \#620.  Mozhgan Pourmoradnasseri is recipient of the Estonian IT Academy Scholarship.  This research is supported by the European Regional Fund through the Estonian Center of Excellence in Computer Science, EXCS.

\appendix
% CV/proofs-alpha.tex  (tamc16)
\section{Proof of Theorem~\ref{thm:alpha}}\label{apx:alpha}
We will assume, for simplicity, that $X=Y=[n]$.

\subsection{Small~$p$: Proof of Theorem~\ref{thm:alpha}~\ref{thm:alpha:small-p}}\label{ssec:alpha:small-p}
We say that a rectangle is \textit{bulky}, if it extends over at least~2 rows and also over at least~2 columns.
The proof of Theorem~\ref{thm:alpha} proceeds by considering three types of rectangles:
\begin{enumerate}[1.]
\item those consisting of exactly one row or column (they give the bound in the theorem);
\item square bulky rectangles;
\item bulky rectangles which are not square.
\end{enumerate}

Let us start with the easiest type (1).  The size of such a rectangle is the number of 1s in the chosen row.
\begin{lemma}
  For all $p,n$, a.a.s., there exists a row in $\Matrix{n}{p}$ containing at least~$pn$ 1s.  If $p\gg (\ln n)/n$, for every constant $\eps \in \lt]0,1\rt]$, a.a.s., no row or column has more than $(1+\eps)pn$ 1s.
\end{lemma}
\begin{proof}
  For the first statement, note that the probability that number of 1s in a fixed row is less than $pn$ is at most $\nfrac12$ (median of a binomial distribution).  Since the rows are independent, the probability that all rows have fewer than $pn$ 1s is at most $2^{-n}$.

  For the second statement, we use an easy Chernoff-type bound (Theorem 4.4(2) in \cite{Mitzenmacher-Upfal:Book}).  Denote by~$X$ the number of 1s in a fixed row of $\Matrix{n}{p}$.  Then
  \begin{equation*}
    \Prb( X \ge (1+\eps)pn ) \le e^{-\eps^2 pn /3} \le e^{-2 \ln n} = n^{-2},
  \end{equation*}
  where the last inequality holds for large enough~$n$, because $pn \gg \ln n$ implies $pn > 6\eps^{-2}\ln n$ for $n$ sufficiently large.  Hence, the probability that a row (or a column) exists which has at least $(1+\eps)pn$ 1s is $o(1)$.
\end{proof}

We now deal with rectangles of type~(2).

\begin{lemma}\label{lem:alpha:small-p:no-square}
  If $p \ge \nfrac5n$, then, a.a.s., there is no square 1-rectangle of size~$\sqrt{pn}\times\sqrt{pn}$.
\end{lemma}
\begin{proof}
  We abbreviate $\kappa := pn$.
  By the union bound, for the probability~$q=q(n)$ that there exists a 1-rectangle of size~$\sqrt{\kappa}\times\sqrt{\kappa}$, we have
  \begin{equation*}
    q \le
    \binom{n}{\sqrt{\kappa}}^2 p^{\kappa}
    \le
    \biggl( \frac{e^2 n}{p} \biggr)^{\sqrt{\kappa}} p^{\kappa}
  \end{equation*}
  Applying $\ln$, we find
  \begin{equation}\label{eq:alpha:small-p:no-square:ln}
    \ln q
    \le
    \sqrt{\kappa} \ln n + 2\sqrt{\kappa} + \sqrt{\kappa} \ln(1/p) - \kappa \ln(1/p)
    =
    \sqrt{\kappa}\Bigl( \ln n + 2 - \bigl(\sqrt{\kappa}-1\bigr)\ln(1/p) \Bigr).
  \end{equation}
  Now we distinguish cases.
  If $(2\ln n)^2/n \le p \le \nfrac1e$, then $\sqrt{\kappa} \ge 2 \ln n$, and hence we can bound the expression in the parentheses in~\eqref{eq:alpha:small-p:no-square:ln} as follows:
  \begin{equation*}
     \ln n + 2 - \bigl(\sqrt{\kappa}-1\bigr)\ln(1/p) \le \ln n + 2 - 2\ln n  +1 \le  - \ln n,
  \end{equation*}
  for all large enough~$n$.  Hence, $q\to 0$ in this region.
  If, on the other hand, $5/n \le p \le (2\ln n)^2/n$, then
  \begin{multline*}
    \ln q
    \le
    \sqrt{5}\Bigl( \ln n + 2 - \bigl(\sqrt{5}-1\bigr)\bigl(\ln n - 2\ln(2\ln n)\bigr) \Bigr)
    \\
    \le
    - \sqrt{5}\bigl( \sqrt{5}-2 \bigr)\ln n + O(\ln\ln n)
    = -\Omega(\ln n).
  \end{multline*}
  Hence, $q\to 0$ in this region, too, which completes the proof of the lemma.
\end{proof}

Finally, we come to rectangles of type~(3).  Consider the probability, $\varrho$, that $\Matrix{n}{p}$ contains a bulky 1-rectangle of size~$s$.
By Lemma~\ref{lem:alpha:small-p:no-square}, if such a 1-rectangle has dimensions $a\times b$, we must have $a < \sqrt{pn}$ or $b < \sqrt{pn}$, or else $\varrho=o(1)$.  We have $\varrho \le 2\varrho'$, where $\varrho'$ is the probability that $\Matrix{n}{p}$ contains a 1-rectangle of size~$s$ consisting of at least as many columns than rows.  For $\varrho'$, we need to consider only 1-rectangles with $a < \sqrt{pn}$.  Moreover, increasing~$b$ if necessary, w.l.o.g., we may restrict to rectangles generated by a row-set of size $a$, with $2\le a \le n$ (the LB~2 comes from the condition that the rectangle be bulky).

\begin{lemma}\label{lem:alpha:small-p:bulky-smallp}
  With $\kappa = pn$, if $5 \le \kappa = O(\polylog n)$, then, a.a.s., there is no bulky rectangle of size at least~$\kappa$.
\end{lemma}
\begin{proof}
  By the remarks above, we have to bound the probability that there exists a row-set of size $a\in\{2,\dots,\sqrt\kappa\}$ which generates a 1-rectangle of size at least $\kappa/a$.

  Firstly, for a given set $K$ of~$a$ rows, we bound the probability that the rectangle it generates has size at least $pn$.
% For that, we use the following well-known inequality \cite[Thm.~1.1]{BollobasBkRndGraphs}: If~$S$ is a $\Bin(n,\varrho)$ r.v., $u>1$, and $m:=\lceil upn\rceil$, then
%   \begin{equation}\label{eq:bollobas}\tag{$*$}
%     \Prb(S\ge m) \le \frac{u}{u-1}\Prb(S=m).
%   \end{equation}
  Denote by~$S$ the number of columns in the rectangle generated by~$K$.  This is a $\Bin(n,p^a)$ r.v.
  and
  % using~\eqref{eq:bollobas} with $u := \frac{\lfloor \kappa/a \rfloor}{p^a n} = 1 +\Omega(1)$,
  we find that
  \begin{equation*}
    \Prb( a\cdot b \ge \kappa)
    =
    \Prb( S \ge \nfrac{\kappa}{a} )
    \le
    \binom{n}{\kappa/a} p^\kappa
    =
    \binom{n}{\kappa/a} \lt( \frac{\kappa}{n} \rt)^\kappa.
  \end{equation*}
  Secondly, we sum over all sets~$K$ of cardinality~$a$, and compute
  \begin{multline*}
    \binom{n}{a} \binom{n}{\kappa/a} p^\kappa (1-p^a)^{n-\kappa/a}
    \le
    n^{a+\kappa/a - \kappa +\kappa\log{n}{\kappa}}
    =
    n^{-\bigl( \kappa(1-\nfrac1a) - a - o(\kappa) \bigr)},
  \end{multline*}
  where $\kappa \log{n}{\kappa} = o(\kappa)$ follows from $\kappa = O(\polylog n)$.

  Now, because $a<\sqrt\kappa$, we have that the exponent on~$\nfrac1n$ is $\kappa(1-\nfrac1a) - o(\kappa) \ge \kappa / 3$, as $a \ge 2$.  Finally, summing over all~$a$, we obtain, as an upper bound for the probability that one of these rectangles has size~$\kappa$ or larger, the expression $n^{-(\kappa/3 -1)}$ which is $o(1)$, as $\kappa \ge 5$.
\end{proof}

For the remaining case, we will need the following numerical fact, whose proof we leave to the reader.
\begin{lemma}\label{lem:alpha:small-p:large-p:chernoff-less-1}
  There exists an $\eps > 0$ such that, for all $p \in \lt]\nfrac18, \nfrac1e\rt]$ and $a\in\{2,3\}$,
  \begin{equation*}
    \biggl( a p^{a-2} \biggr)^{p/a} \biggl( \frac{1-p^a}{1-p/a} \biggr)^{1-p/a} \le 1 - \eps.
  \vspace*{-5ex}%
  \end{equation*}%
  \qed
\end{lemma}

Now we deal with bulky rectangles.
\begin{lemma}\label{lem:alpha:small-p:bulky-largep}
  With  $\kappa := pn$, if $\ln^4 n \le \kappa \le \nfrac{n}{e}$, then, a.a.s., there is no bulky rectangle of size at least~$\kappa$.
\end{lemma}
\begin{proof}
  By the remarks above Lemma~\ref{lem:alpha:small-p:bulky-smallp}, we have to bound the probability that there exists a row-set of size $a\in\{2,\dots,\sqrt{\kappa}\}$ which generates a 1-rectangle of size at least $\kappa/a$.

  For $2 \le a < \sqrt{\kappa}$, let $X_a$ count the number of columns~$y$ with $f(x,y)=1$ for $x=1,\dots,a$.  We are going to show that
  \begin{equation*}
    P := \sum_{a=2}^{\sqrt{\kappa}} \binom{n}{a}  \Prb\bigl( X_a \ge \nfrac{\kappa}{a} \bigr) = o(1).
  \end{equation*}
  The r.v.\ $X_a$ has $\Bin(n,p^a)$ distribution.  We compute
  \begin{multline}\label{eq:alpha:small-p:large-p:prbub}
    \Prb\bigl( X_a \ge \nfrac{\kappa}{a} \bigr)
    \le
    \binom{n}{\nfrac{\kappa}{a}} (p^a)^{\kappa/a}
%    \\
    \le
    \lt( \frac{en}{ \nfrac{\kappa}{a} } \rt)^{\kappa/a}  (p^a)^{\kappa/a}
    =
    \lt(  \lt( (ea)^{1/(a-1)} p \rt)^{\kappa/a}  \rt)^{a-1}.
  \end{multline}
  Now, there exists an constant $\varrho < 1$ such that
  \begin{equation}\label{eq:alph:small-p:large-p:the-rhos}
    (ea)^{1/(a-1)} \le
    \begin{cases}
      8\varrho,     &\text{for all $a\ge 2$, and}\\
      e\varrho,      &\text{for $a \ge 4$.}
    \end{cases}
  \end{equation}
  Consequently, we distinguish two cases:
  \begin{enumerate}[\it (i)]
  \item\label{enum:alpha:small-p:case-1} $p\le \nfrac18$ and
  \item\label{enum:alpha:small-p:case-2} $\nfrac18 < p \le \nfrac{1}{e}$.
  \end{enumerate}

  \paragraph{\it Case~\ref{enum:alpha:small-p:case-1}: $p\le \nfrac18$.}
  In this case, we compute
  \begin{align*}
    P &= \sum_{a=2}^{\sqrt{\kappa}} \binom{n}{a}  \Prb\bigl( X_a \ge \nfrac{\kappa}{a} \bigr)&& \\
    &\le
    \sum_{a=2}^{\sqrt{\kappa}} \binom{n}{a}  \lt(  \lt( (ea)^{1/(a-1)} p \rt)^{\kappa/a}  \rt)^{a-1} && \comment{by \eqref{eq:alpha:small-p:large-p:prbub}}\\
    &\le
    \sum_{a=2}^{\sqrt{\kappa}} \binom{n}{a} \lt( \varrho^{\kappa/a} \rt)^{a-1}&& \comment{by \eqref{eq:alph:small-p:large-p:the-rhos}}\\
    &\le
    \sum_{a=2}^{\sqrt{\kappa}} \binom{n}{a} \lt( \varrho^{\sqrt{\kappa}} \rt)^{a-1}&& \comment{since $a\le \sqrt{\kappa}$ and $\varrho<1$}\\
    &=
    \varrho^{\sqrt{\kappa}} \sum_{a=2}^{\sqrt{\kappa}} \binom{n}{a} \lt( \varrho^{\sqrt{\kappa}} \rt)^{a-2}&&\\
    &\le
    \varrho^{\sqrt{\kappa}} \; n^2 \; \sum_{a=0}^{\sqrt{\kappa}-2} \binom{n-2}{a} \lt( \varrho^{\sqrt{\kappa}} \rt)^a&& \comment{replacing $a\leadsto a-2$}\\
    &\le \varrho^{\sqrt{\kappa}} \; n^2 \; (1+\varrho^{\sqrt{\kappa}})^{n-2} && \comment{Binomial theorem}\\
    &\le \varrho^{\sqrt{\kappa}} \; n^2 \; e^{n \varrho^{\sqrt{\kappa}}} && \\
    &\le \varrho^{\sqrt{\kappa}} \; n^2 \; e^{n \varrho^{\ln^2 n}} && \comment{because $p\ge (\ln^4n)/n$ and $\varrho<1$}\\
    &= o(1) && \comment{because $\varrho = 1-\Omega(1)$.}\\
  \end{align*}

  \paragraph{\it Case~\ref{enum:alpha:small-p:case-2}: $\nfrac18 < p \le \nfrac{1}{e}$.}
  In this case, by~\eqref{eq:alph:small-p:large-p:the-rhos}, the same calculation as in the $p<\nfrac18$-case works if the sum is started with $a=4$.  For the first two terms of the sum, $a=2,3$, we use a Chernoff bound on $X_a$, which gives us (e.g., Eqn.~(2.4) in~\cite{Janson-Luczak-Rucinski:Book})
  \begin{equation}\label{eq:alpha:small-p:large-p:chernoff}
    \Prb\bigl( X_a \ge \nfrac{\kappa}{a} \bigr) \le \lt(   \biggl( a p^{a-2} \biggr)^{p/a} \biggl( \frac{1-p^a}{1-p/a} \biggr)^{1-p/a}   \rt)^n.
  \end{equation}
  Using Lemma~\ref{lem:alpha:small-p:large-p:chernoff-less-1}, we conclude
  \begin{multline*}
    P = \sum_{a=2}^{\sqrt{\kappa}} \binom{n}{a}  \Prb\bigl( X_a \ge \nfrac{\kappa}{a} \bigr)
    \\
    =
    \binom{n}{2} \Prb\bigl( X_2 \ge \nfrac{\kappa}{2} \bigr)
    + \binom{n}{3} \Prb\bigl( X_3 \ge \nfrac{\kappa}{3} \bigr)
    + \sum_{a=4}^{\sqrt{\kappa}} \binom{n}{a}  \Prb\bigl( X_a \ge \nfrac{\kappa}{a} \bigr)
    \\
    = o(1) + o(1) + o(1),
  \end{multline*}
  where the first two ``$o(1)$''s follow from~\eqref{eq:alpha:small-p:large-p:chernoff} and Lemma~\ref{lem:alpha:small-p:large-p:chernoff-less-1}, and the third is the same calculation as in the previous case.
\end{proof}

This concludes the proof of Theorem~\ref{thm:alpha}\ref{thm:alpha:small-p}.

%%%%%%%%%%%%%%%%%%%%%%%%%%%%%%%%%%%%%%%%%%%%%%%%%%%%%%%%%%%%%%%%%%%%%%%%%%%%%%%%%%%%%%%%%%%%%%%%%%%%%%%%%%%%%%%%%%%%%%%%%%%%%%%%%%%%%%%%%%%%%%%%%%%%%%%%%%%%%%%%%%%%%%%%%%%%%%%%%%%%%%%%%%%%%%%%%%%%%%%%
\subsection{Large~$p$: Proof of Theorem~\ref{thm:alpha}\ref{thm:alpha:large-p}}\label{ssec:alpha:large-p}
Now we prove the part of Theorem~\ref{thm:alpha} about $p \ge \nfrac{1}{e}$.  Again, we first prove a statement about square rectangles.
\begin{lemma}\label{lem:alpha:large-p:no-square-rect}
  For every $\eps > 0$ there exists a constant $\lambda_0$ such that, if $n\ge \notp n = \lambda \ge \lambda_0$, then, a.a.s., there is no square 1-rectangle of size
  \begin{equation*}
    \frac{n}{\lambda^{1-\eps}} \times \frac{n}{\lambda^{1-\eps}}
  \end{equation*}
\end{lemma}
\begin{proof}
  This is a direct union bound computation.  With $b := \frac{n}{\lambda^{1-\eps}}$, the probability that such a 1-rectangle exists is at most
  \begin{equation*}
    \binom{n}{b}^{\!\!2} p^{b^2}
    =
    \binom{n}{b}^{\!\!2} (1-\notp)^{b^2}
    \le
    e^{b (2\ln(en/b) - \lambda b/n)}
    =
    e^{b \cdot A_b},
  \end{equation*}
  where
  \begin{align*}
    A_b
    &=
    2 \ln(en/b) - \lambda b/n
    \\
    &=
    2 \ln\bigl(e \, \lambda^{1-\eps} \bigr) - \lambda^\eps
    \\
    &\le -1,
  \end{align*}
  where the last inequality holds if $\lambda \ge \lambda_0$ and $\lambda_0$ is large enough.  The claim follows because $b\to \infty$.
\end{proof}

As above, we need the notion of a ``bulky'' rectangle: Here, we say that a rectangle of dimensions $k\times \ell$ is \textit{bulky,} if $k \le \ell$.  By Lemma~\ref{lem:alpha:large-p:no-square-rect}, in particular, a.a.s., a bulky rectangle must have $k < n/\lambda^{\nfrac23}$.  Again, by exchanging the roles of rows and columns, and multiplying the final probability estimate by~2, we only need to consider 1-rectangles with at least as many columns as rows (i.e., bulky ones).

\begin{proof}[Proof of Theorem~\ref{thm:alpha}\ref{thm:alpha:large-p}]
  For every $b\in[n]$, denote by~$X_b$ the number of columns of the 1-rectangle generated by the row set~$\{1,\dots,b\}$---a random variable with $\Bin(n,p^b)$ distribution.  We prove that, for every $1< u < 2$,
  \begin{equation}\label{eq:alpha:large-p:goal-estimate}
    \sum_{b=1}^{n/\lambda^{\nfrac23}} \binom{n}{b}  \Prb( bX_b \ge u\, ap^an ) = o(1),
  \end{equation}
  which, together with Lemma~\ref{lem:alpha:large-p:no-square-rect}, proves Theorem~\ref{thm:alpha}\ref{thm:alpha:large-p}.

  We split the proof into two lemmas, dealing with the cases $b\le \log{\nfrac1p}e$ and $b \ge \log{\nfrac1p}e$, resp., stated below.  Establishing these lemmas completes the proof of Theorem~\ref{thm:alpha}\ref{thm:alpha:large-p}.
\end{proof}

\begin{lemma}\label{lem:alpha:large-p:b-le-a}
  For every $u\in\lt]1,2\rt[$ there exists a constant $\lambda_0 \ge 1$ such that, for every $\notp \ge \nfrac{\lambda_0}{n}$, and for every $1\le b\le \log{\nfrac1p}e$,
  we have
  \begin{equation*}
    \binom{n}{b} \Prb\Bigl(  X_b \ge u \, \frac{a}{b}p^a n  \Bigr)  = o_u(\nfrac1n).
  \end{equation*}
\end{lemma}

\begin{lemma}\label{lem:alpha:large-p:a-le-b}
  For every $u \in \lt]1,2\rt[$ there exists a constant $\lambda_0$ such that, if $\notp n = \lambda \ge \lambda_0$, and $\log{\nfrac1p}e \le b \le n/\lambda^{\nfrac32}$, then
  \begin{equation*}
    \binom{n}{b} \Prb\Bigl(  X_b \ge u \, \frac{a}{b}p^a n  \Bigr)  = o_u(\nfrac1n).
  \end{equation*}
\end{lemma}

\begin{proof}[Proof of Lemma~\ref{lem:alpha:large-p:b-le-a}]
  Define
  \begin{equation*}
    \delta := \min\biggl( u\,\frac{ap^a}{bp^b} -1 , \; 1 \biggr).
  \end{equation*}
  Note that $\delta \ge u-1 > 0$ by the definition of~$a$ in~\eqref{eq:alpha:large-p:def-a}.  The ``$1$'' on the RHS of the minimum is somewhat arbitrary: the particular version of the Chernoff inequality  which we refer to, \cite[Thm~4.4-2]{Mitzenmacher-Upfal:Book}, requires $\delta \le 1$.
  Using this Chernoff bound in
  \begin{equation*}
    \Prb(X_b \ge u\,ap^an/b) \le \Prb(X_b \ge (1+\delta)\Exp X_b) \le e^{-\delta^2\Exp X_b/3},
  \end{equation*}
  and the inequality $\binom{n}{b} \le (en/b)^b$, we estimate
  \begin{align}
    \ln\lt( \binom{n}{b} \Prb\Bigl(  X_b \ge u\,\frac{a}{b}p^a n  \Bigr) \rt)
    &\le
    b\ln\Bigl(\frac{en}{b}\Bigr) - \delta^2 p^bn/3 \notag\\
    &\le
    n\lt( \frac{b}{n}\ln\Bigl(\frac{en}{b}\Bigr) - \frac{\delta^2}{3e} \rt) &&\comment{since $b\le \log{\nfrac1p}e$}.
    \tag{$*$}\label{eq:alpha:large-p:b-le-a:lnovern}
  \end{align}
  For any real $b\in[1,\log{\nfrac1p}e]$, denote by $A_b$ the term inside the parentheses in~\eqref{eq:alpha:large-p:b-le-a:lnovern}.

  Since $b\mapsto A_b$ is nondecreasing on $[1,n]$, we have, for every $b\in [1,\log{\nfrac1p}e]$,
  \begin{align*}
    A_{b}
    &\le
    A_{\log{\nfrac1p}e} \\
    &\le
    A_{1/\notp}                                                  &&\comment{$A_\cdot$ nondecreasing and $\log{\nfrac1p}e \le \frac{1}{\notp} \le n$, by~\eqref{eq:alpha:a-bounds}}\\
    &\le
    A_{n/\lambda_0}                                              &&\comment{$A_\cdot$ nondecreasing and $1/\notp \le n/\lambda_0 \le n$}\\
    &=
    \frac{  \ln(e\lambda_0) }{ \lambda_0 } - \frac{\delta^2}{3e}  &&\\
    &\le
    \frac{  \ln(e^2\lambda_0) }{ \lambda_0 } - \frac{(u-1)^2}{3e}.&&\comment{as $\delta \ge u-1$.}
  \end{align*}
  Hence, for sufficiently large $\lambda_0$, depending only on~$u$, we have, for all $b \in [1,\log{\nfrac1p}e]$,
  \begin{equation*}
    A_{b} = -\Omega_u(1),
  \end{equation*}
  so that
  \begin{equation*}
    \Prb(X_b \ge ap^an/b) \le e^{-n A_b} = e^{-\Omega_u(n)} = o_u(\nfrac1n)
  \end{equation*}
  which concludes the proof of the lemma.
\end{proof}

\begin{proof}[Proof of Lemma~\ref{lem:alpha:large-p:a-le-b}]
  By Lemma~\ref{lem:alpha:large-p:no-square-rect}, we already know that, if a bulky 1-rectangles generated by~$b$ rows exists with non-$o(1)$ probability, we must have $b < n/\lambda^{\nfrac23}$.

  Define $\delta$ as follows, $0 < u-1 \le \delta := (u-1)\frac{ap^a}{bp^b} \le u\frac{ap^a}{bp^b} -1$, and let
  \begin{equation*}
    \eps :=
    \begin{cases}
      (u-1)^2 /3, &\text{if $\delta \le \nfrac 32$;}\\
      \ln \nfrac52 - 1 + \nfrac25, &\text{otherwise.}
    \end{cases}
  \end{equation*}
  We do the case distinction because we use two slightly different versions of Chernoff in our estimate of 
  \begin{equation*}%\label{eq:alpha:large-p:a-le-b:ln_prb}\tag{P}
    \varrho := \Prb\Bigl( X_b \ge u\frac{a}{b} p^an \Bigr).
  \end{equation*}
  If $\delta \le \nfrac32$, then
  \begin{align*}
    %\hspace*{1em}&\hspace*{-1em}%
    \varrho% \ln\lt(   \binom{n}{b} \Prb\Bigl( X_b \ge u\frac{a}{b} p^an \Bigr)  \rt) &\\
    &\le
    \Prb\Bigl( X_b \ge (1+\delta) \Exp X_b\Bigr)
    \\
    &\le
    e^{- \delta^2 p^b n/3} &&\comment{Chernoff, e.g., \cite[Cor.\ 2.3]{Janson-Luczak-Rucinski:Book}}
    \\
    &\le
    e^{- (u-1)\, (u-1) \frac{a}{b} p^a n/3} &&\comment{definition of $\delta$, and $\delta \ge u-1$}
    \\
    &= 
    e^{- \eps\, \frac{a}{b} p^a n}.
  \end{align*}
  If, on the other hand, $\delta > \nfrac32$, then
  \begin{equation*}
    u \frac{a}{b} p^an
    =
    u\frac{ap^a}{bp^b} \cdot \Exp X_b
    \ge
    (\delta + 1)\cdot \Exp X_b
    \ge
    \tfrac{5}{2} \cdot \Exp X_b,
  \end{equation*}
  and we have, by Eqn.~(2.10) in \cite[Cor.\ 2.4]{Janson-Luczak-Rucinski:Book},
  \begin{equation*}
    \varrho
    \le
    e^{- \eps\, \frac{a}{b} p^a n}.
  \end{equation*}
  In both cases, we conclude
  \begin{align*}
    \hspace*{1em}&\hspace*{-1em}%
    \ln\lt(   \binom{n}{b} \Prb\Bigl( X_b \ge u\frac{a}{b} p^an \Bigr)  \rt) \\
    &\le
    b\ln(en/b) - \eps\, \frac{a}{b} p^a n \\
    &\le
    b\ln(en/b) -  \eps\, \frac{an}{e^2b} &&\comment{$p^a \ge \nfrac1{e^2}$ by~\eqref{eq:alpha:p-to-a-lb}}
    \\
    &\le
    b\ln(en/b) - \eps\, \frac{n^2}{2e^3\lambda b} &&\comment{$a \ge \lfloor \nfrac{n}{e\lambda} \rfloor$ by~\eqref{eq:alpha:a-bounds}, \& $\lfloor \nfrac{n}{e\lambda} \rfloor \ge \nfrac{n}{2e\lambda}$ as $\lambda\le \nfrac{n}{e}$}
    \\
    &\le
    b\ln(e^2\lambda) - \eps\, \frac{n^2}{2e^3\lambda b} &&\comment{as $b\ge \log{\nfrac1p}e \ge \nfrac{n}{e\lambda}$, by~\eqref{eq:alpha:a-bounds}}
    \\
    &\le
    \frac{n}{\lambda^{\nfrac23}}\ln(e^2\lambda)
    -
    \frac{ \eps }{ 2e^3} \frac{n}{\lambda^{\nfrac13}  }    &&\comment{as $b\le n/(\sqrt{\lambda}\ln\lambda)$}
    \\
    &=
    \frac{n}{\lambda^{\nfrac13}} \lt( - \frac{ \eps }{ 2e^3} + o_{\lambda\to\infty}(1) \rt).
  \end{align*}
  Hence, if $\lambda$ is at least a large enough constant, $\lambda_0$, then
  \begin{equation*}
    \binom{n}{b} \Prb\Bigl( X_b \ge u\frac{a}{b} p^an \Bigr) = e^{-\Omega_u(n^{2/3})} = o(\nfrac{1}{n}),
  \end{equation*}
  and the lemma is proven.
\end{proof}

\subsection{Proof of Corollary~\ref{cor:alpha:large-p:1-o1}}
\begin{proof}[Proof of the corollary from Theorem~\ref{thm:alpha}]
 For the given~$p=1-\notp$, if $\nfrac{1}{e}=p^{a}$, we have
 \begin{multline*}
   ap^a
   =
   (1+O(\notp)) \frac{ \log{\nfrac{1}{p}}e }{ e }
   =
   \frac{1+O(\notp)}{e\ln\frac{1}{1-\notp}}
   \\
   =
   \frac{1+O(\notp)}{e\lt( \notp + \notp^2/2 +  \notp^3/3 + \dots\rt)}
   \eqcmt{(*)}
   \frac{1+O(\notp)}{e\notp}
   =
   \frac{1}{e\notp} + O(1)
   =
   \frac{n}{e\lambda} + O(1),
 \end{multline*}
 where equation~(\textasteriskcentered) uses $\notp=o(1)$.  Multiplying by~$n$ and invoking Theorem~\ref{thm:alpha}\ref{thm:alpha:large-p}, we obtain the desired bound.
\end{proof}

% END OF proofs-alpha.tex

% CV/proofs-alpha.tex  (tamc16)
\section{Proof of Theorem~\ref{thm:fool}}\label{apx:fool}
The proof of Theorem~\ref{thm:fool} is extends over the following three subsections.  We first treat upper bounds based on the 1st moment method, then we make the 2nd moment calculation (for the case when $p\to1$ quickly), and finally we show how to obtain fooling sets by combining a matching in random bipartite graphs and a stable set in a random (not bipartite) graph.
\subsection{Upper bounds: The number of fooling sets of size~$r$}\label{ssec:fool:markov}
Let the random variable $X = X_r = X_{r,n,p}$ count the number of fooling sets of size~$r$ in~$f$.  For a set $F \subseteq [n]\times[n]$, denote by $A_F$ the event that~$F$ is a fooling set of~$f$.
We have
\begin{equation}\label{eq:fool:numofool_RV}
  X_r = \sum_F \Ind[ A_F ],
\end{equation}
where the sum ranges over all~$F$ of the form $F = \{ (k_1,\ell_1),\dots,(k_r,\ell_r) \}$, with all the $k_j$'s distinct, and all the~$\ell_j$'s distinct.  There are $r!\,\binom{n}{r}^{\!\!2}$ of these sets~$F$, and hence
\begin{equation*}
  \Exp X_r = r!\,\binom{n}{r}^{\!\!2} \, p^r \, \delta^{^{\binom{r}{2}}}.
\end{equation*}

Elementary calculus shows that, for fixed $r\ge 2$, $p\mapsto r!\,\binom{n}{r}^{\!\!2} \, p^r \, \delta^{^{\binom{r}{2}}}$ is increasing on $[0,1/\sqrt r]$ and decreasing on $[1/\sqrt r,1]$ (see the proof of the \ref{enum:fool:exp_numo:n}-part of Lemma~\ref{lem:fool:exp_numo}).  The following lemma describes for which values of~$r$ the expectation $\Exp X_r$ tends to 0 or infinity, resp., in the relevant range of~$p$.

\newcommand{\fThres}{\textstyle n^{-\nfrac12}\,\sqrt{\ln n}}
\begin{lemma}\label{lem:fool:exp_numo}
  \begin{enumerate}[(a)]
  \item\label{enum:fool:exp_numo:n}%
    If $\displaystyle e/n \; \le \; p \; \le \; \fThres$, %
    then %
    $\displaystyle \Exp X_n \to \infty$.
  \item\label{enum:fool:exp_numo:Cn}%
    For constants $c > 1$, $\eps >0$ %
    if $\displaystyle p \, = \, c \, \fThres$,
    with
 %   \begin{equation}
     $r := (1+\eps)\,\frac{n}{c^2} $
 %   \end{equation}
    we have $\Exp X_r \to 0$.
  \item\label{enum:fool:exp_numo:o_n}%
    If $p \gg \fThres$ and $1-p = \notp \ge n^{-o(1)}$, letting
    \begin{align*}
      r_- := 2 \log{\nfrac1\delta}(pn^2) - 2\log{\nfrac1\delta}\log{\nfrac1\delta}(pn^2)   \text{ and }\\
      r_+ := 2 \log{\nfrac1\delta}(pn^2)
%      \tfrac{r-1}{2}\ln(\nfrac1\delta) &=  \ln(pn^2) - \ln r,
    \end{align*}
    we have $\Exp X_{r_-} \to \infty$, and $\Exp X_{r_+} \to 0$.
  \item\label{enum:fool:exp_numo:const-fool}%
    If $a \in \lt]0,4\rt[$ is a constant and $1-p = \notp = n^{-a}$, then
    $\Exp X_r \to 0$      if $\displaystyle r > \nfrac{4}{a}+1$, and
    $\Exp X_r \to \infty$ if $\displaystyle r < \nfrac{4}{a}+1$.
  \end{enumerate}
\end{lemma}
\begin{proof}\mbox{}\\
\textit{\ref{enum:fool:exp_numo:n}.} %
First of all, we prove that for $r\ge 2$, the function $p\mapsto \Exp X_{r,p}$ is non-decreasing $\lt]0,r^{-\nfrac12}\rt]$ and non-increasing on $\lt[r^{-\nfrac12},1\rt[$.

Clearly, only the function
\begin{equation*}
  f\colon p \mapsto p (1-p^2)^{(r-1)/2}
\end{equation*}
is of interest.  Taking the derivative, we obtain
\begin{equation*}
  f'(p) = (1-p^2)^{(r-1)/2} - (r-1) p^2 (1-p^2)^{(r-3)/2}.
\end{equation*}
If $0<p<1$, then $f'(p) = 0$ and is equivalent to
\begin{equation*}
  0 = 1-p^2 - (r-1) p^2 = 1 - r p^2.
\end{equation*}
For $p < 1/\sqrt{r}$, we have $f'(p) > 0$ and $p > 1/\sqrt{r}$, we have $f'(p) < 0$.

Now, for $p=e/n$, using Stirling's formula, we have
\begin{equation*}
  \Exp X_n
  =
  n! \lt(\frac{e}{n}\rt)^n \lt( 1-\frac{e^2}{n^2}\rt)^{\binom{n}{2}}
  = \Theta( \sqrt{n} ),
\end{equation*}
so $\Exp X_n$ tends to infinity with $n\to\infty$.

Finally, let $p = \fThres$.  We have
\begin{multline*}
  \frac{\ln \Exp X_n}{n}
  =
  \frac{ \ln\bigl(  n! \, p^n \, \delta^{\binom{n}{2}} \bigr) }{n}
  =
  -1+ o(1) + \ln n - \ln(\nfrac1p) - \tfrac{n-1}{2} \ln(\nfrac1\delta)
  \\
  \ge
  -1+ o(1) + \ln n - \ln(\nfrac1p) - \tfrac{n-1}{2} p^2,
\end{multline*}
where we used $\ln(\nfrac1\delta) = \ln(1/(1-p^2)) \le p^2 + O(p^4)$ and $np^4 = o(1)$ in the last inequality.  Replacing $p$, we get
\begin{equation*}
  \frac{\ln \Exp X_n}{n}
  \ge
  % \ln n
  % - \tfrac12\ln n +
  \tfrac12\ln\ln n
  % - \tfrac12\ln n
  + O(1),
\end{equation*}
which proves the claim in~\ref{enum:fool:exp_numo:n} for this particular value of~$p$.  

\mypar%
  \textit{\ref{enum:fool:exp_numo:Cn}.} %
  First of all, note that, for $4 \le r < n$, using the estimates %\todo[??]
  \begin{eqnarray*}
    \sqrt{r} \, \Bigl( \frac{r}{e} \Bigr)^r         \le&\displaystyle r! &\le r\Bigl( \frac{r}{e} \Bigr)^r, \text{ and}\\
    \tfrac{1}{3\sqrt r} \, e^{r-r^2/(n-r)} \Bigl( \frac{n}{r} \Bigr)^r \le&\displaystyle \binom{n}{r} &\le
    e^r \Bigl( \frac{n}{r} \Bigr)^r,
  \end{eqnarray*}
  we have
  \begin{multline}\label{eq:fool:bounds_for_Exp_fool}\tag{$*$}
    1 - \tfrac{r}{n-r} - O(\tfrac{\ln r}{r})
    \le
    \\
    \frac{ \ln\bigl(  r! \binom{n}{r}^2\, p^r \, \delta^{\binom{r}{2}} \bigr) }{r}
    -
    \biggl(
    \ln(n^2) - \ln(\nfrac1p) - \tfrac{r-1}{2}\ln(\nfrac1\delta) - \ln r
    \biggr)
    \\
    \le
    1 + \tfrac{\ln r}{r}.
  \end{multline}
  (We will use this for~\ref{enum:fool:exp_numo:o_n}, too.)

  Now, with $c>1$, $p = c\,\fThres$ and $r = (1+\eps) n / c^2 = (1-\Omega(1))n$, we get
  \begin{align*}
    \frac{\ln \Exp X_r}{r}
    &=
    \ln(n^2) - \ln(\nfrac1p) - \tfrac{r-1}{2}\ln(\nfrac1\delta) - \ln r
    + O(1)
    \\
    &=
    \ln(n^2)
    - \tfrac12 \ln(n/\ln n)
    - \tfrac{r-1}{2}\ln(\nfrac1\delta)
    - \ln n
    +
    O(1)
    \\
    &=
    \tfrac12 \ln n
    - \tfrac{r-1}{2}n\ln(\nfrac1\delta)
    + O( \ln\ln n )
    \\
    &=
    \tfrac12 \ln n
    - \tfrac{r-1}{2}\bigl( p^2 + O(p^4) )
    + O(\ln\ln n)
    \\
    &=
    - \tfrac{\eps}{2} \ln n
    + O(\ln\ln\ln n),
  \end{align*}
  which proves $\Exp X_r \to 0$.

  \mypar%
  \textit{\ref{enum:fool:exp_numo:o_n}.} %
  With $r := r_+ = 2 \ln(pn^2)/\ln(\nfrac1\delta)$, using the upper bound from~\eqref{eq:fool:bounds_for_Exp_fool}, we get
  \begin{align*}
    \frac{\ln \Exp X_r}{r}
    &\le
    \ln(pn^2)
    - \tfrac{r-1}{2}\ln(\nfrac1\delta)
    - \ln r
    + 1 + \tfrac{\ln r}{r}\\
    &=
    \tfrac12 \ln(\nfrac1\delta)
    - \ln r
    + 1 + \tfrac{\ln r}{r}\\
    &= -\Omega(1),
  \end{align*}
  where the last equation follows from $r \to \infty$ (due to $\notp \ge n^{-o(1)}$), which also implies $\Exp X_r \to 0$.

  On the other hand, with $r := r_- = 2 \log{\nfrac1\delta}(pn^2) - 2\log{\nfrac1\delta}\log{\nfrac1\delta}(pn^2)$, using the upper bound from~\eqref{eq:fool:bounds_for_Exp_fool}, we get
  \begin{align*}
    \frac{\ln \Exp X_r}{r}
    &\ge
    \ln(pn^2)
    - \tfrac{r-1}{2}\ln(\nfrac1\delta)
    - \ln r
    + 1 -O(\tfrac{\ln r}{r})\\
    &\ge
    \bigl( \log{\nfrac1\delta}\log{\nfrac1\delta}(pn^2) \bigr) \ln(\nfrac1\delta)
    - \ln r
    + 1 + O(\tfrac{\ln r}{r})\\
    &=
    \ln \log{\nfrac1\delta}(pn^2)
    - \ln r
    + 1 + O(\tfrac{\ln r}{r})\\
    &\ge
    -\ln 2 + 1 + O(\tfrac{\ln r}{r})\\
    &= \Omega(1).
  \end{align*}
  Again, the last equation and the conclusion $\Exp X_r \to \infty$ follows from $r\to\infty$.

  \mypar%
  \textit{\ref{enum:fool:exp_numo:const-fool}.} %
  Finally, let $0 < a < 4$ be a constant and $1-p = \notp = n^{-a}$.  Noting that $\delta = (1+p)\notp = \Theta(\notp)$, if $r=O(1)$, we have
  \begin{align*}
    \Bigl( \Exp X_r \Bigr)^{\nfrac1r}
    = \Theta\bigl( n^2 \notp^{(r-1)/2} \bigr)
    = \Theta\bigl( n^{2-a(r-1)/2} \bigr),
  \end{align*}
  which implies $\Exp X_r \to \infty$ if $\displaystyle r > \nfrac{4}{a}+1$, and $\Exp X_r \to 0$ if $\displaystyle r < \nfrac{4}{a}+1$.
\end{proof}

From this lemma, we immediately get the upper bound on $\fool(f)$ in Theorem~\ref{thm:fool}\ref{thm:fool:ub}.

\begin{proof}[Proof of Theorem~\ref{thm:fool}\ref{thm:fool:ub}]
  Follows from~\ref{enum:fool:exp_numo:o_n}.
\end{proof}

Item~\ref{enum:fool:exp_numo:n} of the lemma suggests the question, for which $p$ the value of $\fool(f)$ drops from $(1-o(1))n$ to $(1-\Omega(1))n$.  If the expectation is ``right'', this happens crossing from $p=\sqrt{(\ln n)/n}$ to $p=(1+\eps)\sqrt{(\ln n)/n}$.  This is supported by the fact that our lower bounds in this region---in the next subsection---appear to be quite simple, in that they only consider one fixed maximal matching in $\Hbipfnp$, and delete edges from it until it becomes cross free.

\subsection{Second moment calculation}\label{ssec:fool:2nd_moment}
\begin{lemma}\label{fool:variance}
  If $r=O(1)$ and $p\delta \gg \nfrac1n$, then $\displaystyle \Var(X_r) = o\bigl( \bigl(\Exp X_r\bigr)^{\!2\,} \bigr)$.
\end{lemma}
\begin{proof}
  With the notations as in equation~\eqref{eq:fool:numofool_RV}, let $F_0 := \{ (1,1),\dots,(r,r) \}$, and abbreviate $A_0 := A_{F_0}$.  We have
  \begin{equation*}
    \Exp(X^2) = \Exp X \cdot \sum_{F} \Prb( A_F \mid A_0 )
  \end{equation*}
  where the sum ranges over all~$F$ of the form $F = \{ (k_1,\ell_1),\dots,(k_r,\ell_r) \}$, with all the $k_j$'s distinct, and all the~$\ell_j$'s distinct, as in~\eqref{eq:fool:numofool_RV}.

  If $F\subset \{r+1,\dots,n\}\times\{1,\dots,n\}$, % or $F \subset \{1,\dots,n\}\times\{r+1,\dots,n\}$,
  then the events $A_F$ and~$A_0$ are clearly independent, so that, with the following sum ranging over these~$F$, we have
  \begin{equation*}
    \sum_{F} \Prb( A_F \mid A_0 ) = \frac{ (n-r)_r }{ (n)_r }  \Exp X.
  \end{equation*}
  Consequently, we have
  \begin{equation*}
    \Exp(X^2) = \frac{ (n-r)_r }{ (n)_r } \bigl( \Exp X \bigr)^2  + \Exp X \cdot \sum_{F} \Prb( A_F \mid A_0 ),
  \end{equation*}
  where the last sum ranges over all~$F$ with %both
  $F \cap \{1,\dots,r\}\times\{1,\dots,n\} \ne \emptyset$.  %and $F \cap \{1,\dots,n\}\times\{1,\dots,r\} \ne \emptyset$
  For each such~$F$,
  \begin{equation*}
    \Prb( A_F \mid A_0 ) = O\biggl( \frac{1}{p\delta n} \biggr)^{O(r^2)} \Prb( A_F ),
  \end{equation*}
  with absolute constants in the big-$O$s.

  Hence, if $r=O(1)$ and $p\delta \gg \nfrac1n$,
  \begin{equation*}
    \Exp(X^2)
    =
    \frac{ (n)_r }{ (n-r)_r } \bigl( \Exp X \bigr)^2
    +
    O\biggl( \frac{1}{p\delta n} \biggr)^{O(r^2)}\bigl( \Exp X \bigr)^2
    =
    (1+o(1)) \bigl( \Exp X \bigr).
  \end{equation*}
  This proves the statement of the lemma.
\end{proof}

\begin{proof}[Proof of Theorem~\ref{thm:fool}\ref{thm:fool:notp-small}]
  The upper bound, for general~$a$ is in Lemma~\ref{lem:fool:exp_numo}\ref{enum:fool:exp_numo:const-fool}.  The lower bound when $a<1$ follows from Lemma~\ref{lem:fool:exp_numo}\ref{enum:fool:exp_numo:const-fool} and Lemma~\ref{fool:variance}.
\end{proof}

\newcommand{\xfmatch}{\nu^\times}
\subsection{Lower bounds: Cross-free sub-matchings}\label{ssec:fool:bipartite-matching}
Let $\xfmatch(\cdot)$ denote the size largest cross-free matching of a bipartite graph.

Let $H$ be a bipartite graph, and $m=\{e_1,\dots,e_r\} \subseteq E(H)$ a matching in~$H$.  Define the graph~$G'=G'(H,m)$ with vertex set $V(G') = \{1,\dots,r\}$ and $\{k,\ell\}\in E(G')$ if $e_k$, $e_\ell$ induce a $K_{2,2}$ in~$H$.    Then $\displaystyle \xfmatch(H) \ge \alpha(G')$ holds: for any stable set~$A$ of $G'$, the set $\{ e_j \mid j \in A\}$ is a cross-free matching in~$H$.

Our strategy for obtaining a large cross-free matching will be this: fix a large matching~$m$ in $\Hbipfnp$, then find a large stable set in the corresponding random graph $G'_{n,p}(m) := G'(\Hbipfnp,m)$.  This random graph behaves similarly to an Erd\H{o}s-Renyi random graph with $\abs{m}$ vertices and edge-probability $p^2$.  The following technical lemma takes care of the dependency issues which arise.

Let $\Gnp{r}{q}$ denote the Erd\H{o}s-Renyi random graph with~$r$ vertices and edge probability~$q$.
\begin{lemma}\label{lem:fool:bipartite}
  For all positive integers~$n,r,a$, and $p\in[0,1]$, we have
  \begin{equation*}
    \Prb\Bigl( \xfmatch(\Hbipfnp) < a \quad\&\quad  \nu(\Hbipfnp) \ge r \Bigr) \quad\le\quad \Prb\bigl( \alpha(\Gnp{r}{p^2}) < a \bigr).
  \end{equation*}
\end{lemma}
\newcommand{\thematchings}{\mathcal M}
\begin{proof}
  Let $\thematchings$ be the set of matchings of size~$r$ of $K_{n,n}$, and for each $m\in\thematchings$ denote by $C_m$ the event that $\Hbipfnp$ contains~$m$.  Fix a matching $m\in\thematchings$.  For every edge~$e\in E(K_{n,n})$, we have
  \begin{equation*}
    \Prb\bigl( e \in \Hbipfnp \mid C_m \bigr) = p,
  \end{equation*}
  and these events are jointly independent.
  Hence, for each potential edge $e'$ of $G'_{n,p}(m)$,
  \begin{equation*}
    \Prb\bigl( e' \in G'_{n,p}(m) \mid C_m \bigr) = p^2,
  \end{equation*}
  again with joint independence of the events.

  Now, denote by~$A$ the event that there does not exists a cross-free matching of size larger than~$a$ in $\Hbipfnp$.  By the discussion above, $A$ and $C_m$ together imply $\alpha(G'_{n,p}(m)) < a$, so that
  \begin{equation*}
        \Prb\bigl( A                         \mid C_m \bigr)
    \le \Prb\bigl( \alpha(G'_{n,p}(m)) < a \mid C_m \bigr)
    =   \Prb\bigl( \alpha(\Gnp{r}{p^2}) < a \bigr).
  \end{equation*}
  It follows that
  \begin{multline*}
    \Prb\Bigl( \xfmatch(\Hbipfnp) < a \quad\&\quad  \nu(\Hbipfnp) \ge r \Bigr)
    = \Prb\bigl( A \cap \bigcup_m C_m \bigr)
    \le \sum_m \Prb\bigl( A \cap C_m \bigr)
    \\
    = \sum_m \Prb\bigl( A \mid C_m \bigr)\Prb(C_m)
    \le \Prb\bigl( \alpha(\Gnp{r}{p^2}) < a \bigr),
  \end{multline*}
  which concludes the proof of the lemma.
\end{proof}

\begin{remark}\label{rem:fool:use-lem:fool:bipartite}
  We will use Lemma~\ref{lem:fool:bipartite} in the following way: If $p$, $r_-$, $r_+$ are such that both
  \begin{equation}\label{eq:fool:scholie-lem-bipartite-cond}
    \begin{aligned}
      \Prb\bigl( \nu(\Hbipfnp) < r_+\bigr)           &= o(1), \text{ and}\\
      \Prb\bigl( \alpha(\Gnp{r_+}{p^2}) < r_- \bigr) &= o(1),
    \end{aligned}
  \end{equation}
  then, a.a.s., $f$ has a fooling set of size~$r_-$.  Indeed,
  \begin{align*}
    &\Prb\bigl( \fool(f) < r_- \bigr)
    \\
    &\le \Prb\Bigl( \xfmatch(\Hbipfnp) < r_- \quad\&\quad \nu(\Hbipfnp) \ge r_+ \Bigr) + \Prb\bigl( \nu(\Hbipfnp) < r_+ \bigr)
    \\
    &\le \Prb\bigl( \alpha(\Gnp{r_+}{p^2}) < r_- \bigr) + \Prb\bigl( \nu(\Hbipfnp) < r_+ \bigr) &&\comment{Lemma~\ref{lem:fool:bipartite}}
    \\
    &= o(1)+o(1) && \comment{by~\eqref{eq:fool:scholie-lem-bipartite-cond}}.
  \end{align*}
\end{remark}

We are now ready to prove the first two items of Theorem~\ref{thm:fool}.  We start with the easiest part.

\begin{proof}[Proof of Theorem~\ref{thm:fool}\ref{thm:fool:alpha}]
  This is a direct consequence of the remark with $r_- := a(p^2)$ and $r := n$, since, if $pn - \ln n \to \infty$, then $\nu(\Hbipfnp) = n$, a.a.s. (e.g., \cite[Thm~4.1]{Janson-Luczak-Rucinski:Book}).
\end{proof}

\begin{proof}[Proof of Theorem~\ref{thm:fool}\ref{thm:fool:o-sqrtn}]
  Let $\eps > 0$ be a constant.  Proceeding as in Remark~\ref{rem:fool:use-lem:fool:bipartite}, with $r_- := r$ and $r_+ := (1+\eps)r$, if both a.a.s.\ $\nu(\Hbipfnp) \ge r$ and a.a.s.\ $\alpha(\Gnp{r}{p^2}) \ge (1-\eps)r$, then, a.a.s.,
  \begin{equation*}
    (1-\eps) \nu(\Hbipfnp) \le \fool(f) \le \nu(\Hbipfnp).
  \end{equation*}
  Letting $\eps$ tend to~0 then gives the desired result.

  For $n^{-3/2} \le p=o(n)$, a.a.s., the number of edges of $\Gnp{n}{p^2}$ is $o(1)$, and hence $\alpha(\Gnp{n}{p^2}) = (1-o(1))n$, while easy arguments show that a.a.s.\ $\nu(\Hbipfnp) = \Omega(n)$ with concentration in a window of size $O(\sqrt n)$.   Hence the conditions~\eqref{eq:fool:scholie-lem-bipartite-cond} are satisfied.

  For $p=\Omega(1/n)$, a classical result by Karp \& Sipser~\cite{KarpSipser81:match} states that there is a function $h\colon]0,\infty[\to[0,1]$ with $\lim_{c\to\infty}h(c)=1$ such that if $p=c/n$, then, a.a.s., $\nu(\Hbipfnp) = (1-o(1))h(c)/n$.  Since $p=o(1/\sqrt n)$, a.a.s., the number of edges of $\Gnp{n}{p^2}$ is $o(n)$, and hence $\alpha(\Gnp{n}{p^2}) = (1-o(1))n$.  It follows that $\fool(f) = (1-o(1))\nu(\Hbipfnp)$.  In particular, if $p \gg 1/n$, then, a.a.s, $\nu(\Hbipfnp) = (1-o(1))n$.
\end{proof}

% END OF proofs-fool.tex

% CV/proofs-alpha.tex  (tamc16)
\section{Proofs for Section~\ref{sec:ndcc}}\label{apx:ndcc}
\subsection{The ``usual calculation''}\label{ssec:apx:chi:swift-caculation}
With
\begin{equation*}
  \alpha := \max\biggl( 2\lambda, \ \frac{ (1+\eps)\ln  n}{  \ln\bigl( \frac{\ln n}{e\lambda} \bigr) } \biggr),
\end{equation*}
we have to show that
\begin{equation*}
  \alpha\ln(\alpha/e\lambda) \ge \ln n.
\end{equation*}
We write it down informally. In the following list of inequalities, the each one is implied by the next one:
\begin{align*}
  \alpha\ln(\alpha/e\lambda) &\ge \ln n &&\comment{replace $\alpha$ by the 2nd term in the max}\\
  (1+\eps)\frac{\ln\lt(\frac{\alpha}{e\lambda}\rt)}{ \ln\lt( \frac{\ln n}{e\lambda} \rt) }&\ge 1\\
  \alpha &\ge  \ln^{1/(1+\eps)} n\\%  &&\comment{replace $\alpha$ by the 2nd term in the max}\\
  \frac{ (1+\eps)\ln  n}{  \ln\bigl( \frac{\ln n}{e\lambda} \bigr) } &\ge  \ln^{1/(1+\eps)} n &&\comment{is true.}
\end{align*}

\subsection{Chernoff}\label{ssec:apx:chi:binomial-ub}
We have no good reference for the following simple Chernoff estimate (it is almost exactly Theorem~5.4 in~\cite{Mitzenmacher-Upfal:Book}, except that we allow $\lambda\to\infty$ slowly).  For the sake of completeness, we include it here.
\begin{lemma}\label{lem:binomial-ub}
  Let $\notp = \lambda/n$ with $1 < \lambda = o(n)$, and $2\lambda \le \alpha \le n/2$.  The probability that a $\Bin(n,\notp)$ random variable is at least $\alpha$ is at most
  \begin{equation}
    O\bigl(\nfrac{1}{\sqrt \alpha}\bigr) \cdot e^{-\lambda} \Bigl( \frac{e\lambda}{\alpha} \Bigr)^\alpha.
  \end{equation}
\end{lemma}
\begin{proof}[Proof of Lemma~\ref{lem:binomial-ub}]
  Using Thm~1.1 in~\cite{BollobasBkRndGraphs} (here we need the $\alpha \ge 2\lambda$), and the usual estimates for binomial coefficients, we find that said probability (for~$n$ sufficiently large) is at most an absolute constant times
  \begin{multline*}
    \Prb\Bigl( \Bin(n,\notp) = \alpha \Bigr)
    \le
    \frac{1.1}{\sqrt{2\pi \alpha (n-\alpha)/n}} \Bigl( \frac{\lambda}{\alpha} \Bigr)^\alpha \Bigl( \frac{n-\lambda}{n-\alpha} \Bigr)^{n-\alpha}
    \\
    \le
    \frac{1}{\sqrt \alpha} \Bigl( \frac{\lambda}{\alpha} \Bigr)^\alpha \Bigl( 1 - \frac{\alpha-\lambda}{n-\alpha} \Bigr)^{n-\alpha}
    \le
    \frac{1}{\sqrt \alpha} \Bigl( \frac{\lambda}{\alpha} \Bigr)^\alpha e^{\alpha-\lambda},
  \end{multline*}
  as promised.
\end{proof}

\subsection{Number of distinct rows}\label{ssec:apx:chi:distinct-rows}
\begin{proof}[Proof of Lemma~\ref{lem:chi:numo_distinct_rows}]
  For notational convenience, for $k=1,\dots,n$, let
  \begin{equation*}
    S_k := \{ \ell \mid M_{k,\ell} = 0 \}
  \end{equation*}
  The $S_k$ are random sets, where the events $\ell \in S_k$ are all independent and have probability~$\notp$.
  For $m \ge 0$, with $\Zero := \{1,\dots,n\}$, denoting by
  \begin{equation*}
    X_m := \abs{ \{ S_1,\dots,S_k \} \setminus \{\Zero\} },
  \end{equation*}
  the number of distinct non-zero rows among the first~$m$ rows of $\Matrix{n}{p}$, we need to show that $X_n = \Omega(n)$.
  This is quite easy for $\notp = \Omega(\nfrac1n)$, i.e., Item~\ref{enum:apx-chi:numo_distinct_rows:large_p}.  Here, we just prove it in the case that $\notp \le 1/2n$, i.e., Item~\ref{enum:apx-chi:numo_distinct_rows:small_p}.

  Denote by $A_{m+1}$ the event that the $(m+1)$st row is zero or a duplicate of the first~$m$ rows, i.e., that
  \begin{equation*}
    S_{m+1} \in \{\Zero,S_1,\dots,S_m\}.
  \end{equation*}
  We enumerate the distinct sets: $\{S_1,\dots,S_m\} =: \{S_{k_1},\dots,S_{k_{X_m}}\}$.  Now, for $m\ge 1$, we have
  \begin{multline*}
    \Prb\Bigl( A_{m+1} \Bigm|   \abs{S_1},\dots,\abs{S_m}, X_m \Bigr)
    %\\
    \shoveright{%
      =
      \Prb\Bigl( S_{m+1} \in \{\Zero,S_1,\dots,S_m\}    \Bigm|    \abs{S_1},\dots,\abs{S_m}, X_m\Bigr)
    }\\\shoveright{%
    =
    \Prb(S_{m+1} = \Zero) + \sum_{j=1}^{X_m} \Prb\Bigl( S_{m+1} = S_{k_j} \Bigm|  \abs{S_1},\dots,\abs{S_m},  X_m \Bigr)
    }\\%\shoveright{%
    =
    \notp^n + \sum_{j=1}^{X_m} \notp^{\sabstight{S_{k_j}}} p^{n-\sabstight{S_{k_j}}}
    % }\\
    \le
    \notp^n + p^n + \max(0,X_m-1)\notp \onespace p^{n-1},
  \end{multline*}
  where the last inequality comes from the fact that, since the $S_{k_j}$ are all distinct, at most one of them has cardinality~0.
  Hence, for $m\ge 2$,
  \begin{align*}
    \Prb\bigl( A_{m+1} \bigm|    X_m, X_1=1 \bigr)
    &\le
    \notp^n + p^n + (X_m-1)\notp p^{n-1}
    \\
    &\le
    \notp^n + p^n -\notp p^{n-1} + \notp p^{n-1} X_m.
  \end{align*}
  Now, for $m \ge 1$,
  \begin{multline*}
    \Exp\bigl( X_{m+1} \bigm | X_m, X_1=1 \bigr)
    =
    X_m + 1 - \Prb(A_{m+1}\mid X_m, X_1=1),
    \\
    \ge
    X_m + 1 - \notp^n - p^n + \notp p^{n-1} - \notp p^{n-1} X_m
    \\
    =
    1 + \notp p^{n-1} - \notp^n - p^n + (1-\notp p^{n-1}) X_m.
  \end{multline*}
  Using the law of total probability and solving the recursion\footnote{%
    The recursion: $\displaystyle \mu_{m+1} = \alpha + \beta\mu_m = \ldots = \alpha \sum_{j=0}^{m-1} \beta^j + \beta^m \mu_1 = \alpha\frac{1-\beta^m}{1-\beta} + \beta^m \mu_1$.
  }, %
  we find that
  \begin{equation*}
    \Exp\bigl( X_m \bigm | X_1=1 \bigr)
    \ge
    (1 + \notp p^{n-1} - \notp^n - p^n) \frac{1 - (1-\notp p^{n-1})^{m-2}}{\notp p^{n-1}}
    + (1-\notp p^{n-1})^{m-1}
  \end{equation*}
  With $\lambda := \notp n$, again, note that, since, by our assumption above, $\lambda \le \nfrac12$, using the Bernoulli inequalities $1-tn \le (1-t)^n \le 1-tn+t^2\binom{n}{2}$ for $t<1$, we have
  \begin{equation*}
    \frac12 \le 1 - \lambda \le p^n \le p^{n-1} \le 1 - \lambda\lt( \frac{n-1}{n} + \lambda\frac{n-1}{n}\rt) \le 1,
  \end{equation*}
  so that
  \begin{equation*}
    (1-\notp p^{n-1})^{m-2} \le (1-\nfrac{\notp}{2})^{m-2} \le 1   -  \frac{\lambda}{2}\lt( \frac{m-2}{n} + \frac{\lambda}{2}\frac{m-2}{n} \rt).
  \end{equation*}
  We conclude that, for $m=n$,
  \begin{multline*}
    \Exp\bigl( X_m \bigm | X_1=1 \bigr)
    \ge
    (1  - p^n) \frac{1 - (1-\notp p^{n-1})^{m-2}}{\notp p^{n-1}}
    \\
    \ge
    \lambda \lt( \frac{n-1}{n} + \lambda\frac{n-1}{n}\rt) \cdot \frac{   \frac{\lambda}{2} \lt( \frac{m-2}{n} + \frac{\lambda}{2}\frac{m-2}{n} \rt)  }{ \lambda/n }
    \ge (1+o(1)) \frac{\lambda n}{2}.
  \end{multline*}
  Since $\Prb(X_1=1) = \Prb(S_1=\Zero) = (1-\notp^n) = 1-o(1)$, this implies $\Exp X_n \ge \Exp(X_n\mid X_1=1)\Prb(X_1=1) \ge (1-o(1))\nfrac{\lambda n}{2}.$

  To obtain the a.a.s.\ statement from the one about the expectation, we use the usual Martingale-based concentration bound (Corollary 2.27 in~\cite{Janson-Luczak-Rucinski:Book}): as changing one row can affect $X_n$ by at most~1, we get
  \begin{equation*}
    \Prb\bigl(  X_n \le \nfrac{\lambda n}{4} \bigr)
    \le
    \Prb\bigl(  X_n \le \Exp X_n - \nfrac{\lambda n}{4} \bigr)
    \le
    e^{-\nfrac{(\lambda n)^2}{32n}}
    =
    e^{-\Omega(\lambda^2 n)}
    = o(1),
  \end{equation*}
  where the last equation follows from the condition $n^{-\nfrac32} = o(\notp)$.
\end{proof}

% END OF proofs-ndcc.tex

\end{document}